\documentclass[journal]{IEEEtran}
\usepackage{multicol}
\usepackage{indentfirst}
\usepackage[pdftex]{graphicx}
\usepackage{amsthm}
\usepackage{amsmath}
\usepackage{amssymb}
\usepackage{amsfonts}
\usepackage{mathrsfs}
\usepackage{graphicx}
\usepackage{enumerate}
\usepackage{multirow}
\usepackage{booktabs}
\usepackage{enumitem}
\usepackage{enumerate}
\usepackage[noend]{algpseudocode}
\usepackage{algorithmicx,algorithm}
\usepackage{bm}

\newtheorem{Theorem}{Theorem}
\newtheorem{Proposition}{Proposition}

\newtheorem*{Remark}{Remark}
\newtheorem*{Remarks}{Remarks}
\newtheorem*{Square Root Law}{Square Root Law}
\newtheorem{Lemma}{Lemma}
\newtheorem{Fact}{Fact}

\ifCLASSINFOpdf
\else
\fi

\begin{document}
\title{Total Variation Distance Based Performance Analysis of Covert Communication over AWGN Channels in Non-asymptotic Regime}

\author{Xinchun~Yu,
  Shuangqin Wei and Yuan~Luo  % stops a space
        \thanks{This work was supported by China Program of International S\&T Cooperation 2016YFE0100300.}
\thanks{Xinchun Yu and Yuan Luo are with the School of Electronic Information and Electrical Engineering, Shanghai Jiao Tong University, Shanghai 200240, China. Luo is the corresponding author. (e-mail: moonyuyu@sjtu.edu.cn; yuanluo@sjtu.edu.cn).

 Shuangqing Wei is with the Division of Electrical and Computer Engineering
School of Electrical Engineering and Computer Science,
Louisiana State University, Baton Rouge, LA 70803, USA (e-mail: swei@lsu.edu).}
% <-this % stops a space
}

\maketitle



% \fi
\IEEEpeerreviewmaketitle


\begin{abstract}
This paper investigates covert communication over an additive white Gaussian noise (AWGN) channel in finite block length regime on the assumption of Gaussian codebooks.  We first review some achievability and converse bounds on the throughput under maximal power constraint. From these bounds and the analysis of TVD at the adversary, the first and second asymptotics of covert communication are investigated by the help of some divergences inequalities. Furthermore, the analytic solution of TVD, and approximation expansions which can be easily evaluated with given $\bm{snr}$ (signal noise ratio) are presented. In this way, the proper power level for covert communication can be approximated with given covert constraint of TVD, which leads to more accurate estimation of the power compared with preceding bounds. Moreover, the connection between Square Root Law and TVD is disclosed to be on the numerical properties of incomplete gamma functions. Finally, the convergence rates of TVD for $\bm{snr} = n^{-\tau}$ with $\tau > 0.5$ and $\tau < 0.5$ are studied when the block length tends to infinity, which extends the previous extensively focused work on $\tau = 0.5$. Further elaboration on the effect of such asymptotic characteristics on the primary channel's throughput in finite block regime is also provided. The results will be very helpful for understanding the behavior of the total variation distance and practical covert communication.
\end{abstract}
\begin{IEEEkeywords}
Covert communication, finite block length,  metric of discrimination, total variance, convergence rate.
\end{IEEEkeywords}
\section{Introduction}
Security is very important aspect of wireless communication. Covert communication or communication with low probability of detection (LPD), where it is required that the adversary should not learn whether the legitimate parties are communicating nor not, has been studied in a lot of recent works. Typical  scenarios arise in underwater acoustic communication \cite{Roee Diamant} and dynamic spectrum access in wireless channels, where secondary users attempt to communicate without being detected by primary users or users wish to avoid the attention of regulatory entities \cite{Matthieu R}. The information theory for covert communication was first characterized on AWGN channels in \cite{Boulat A} and DMCs in \cite{Matthieu R}\cite{Ligong Wang}, and later in \cite {Pak Hou Che} and \cite{Abdelaziz} on BSC and MIMO AWGN channels, respectively. It has been shown that the throughput of covert communication follows the following square root law (SRL) \cite{Boulat A}.
\begin{Square Root Law}
 In covert communication,  for any $\varepsilon > 0$, the transmitter is able to transmit $O(\sqrt{n})$ information bits to the legitimate receiver by $n$ channel uses while lower bounding the adversary's sum of probability of detection errors $\alpha + \beta \geq 1- \varepsilon$ if she knows a lower bound of the adversary's noise level ($\alpha$ and $\beta$ are error probabilities of type I and type II in the adversary's hypothesis test). The number of information bits will be $o(\sqrt{n})$ if she doesn't know the lower bound.
\end{Square Root Law}

If $\theta_n$ is denoted to be the $\bm{snr}$ at the main channel, then the maximal number of information bits that can be transmitted by $n$ channel uses is $\frac{1}{2}n\log(1+\theta_n)$  over AWGN channels when the input distribution is Gaussian. From the Square Root Law, the maximal number of information bits by $n$ channel uses is $O(\sqrt{n})$ if a lower bound of the adversary's noise is known, hence we have
$n\theta_n = \omega(1)$,
that is: there exists a constant $n_0 > 0$ such that $1 < n\theta_n$ for any $n \geq n_0$.
Furthermore, we have
$1 \leq \frac{1}{2}n\log(1+\theta_n) = \frac{\frac{1}{2}n\ln(1+ \theta_n)}{\ln2} \sim \frac{n\theta_n}{2\ln2}\sim O(\sqrt{n}).$
The first inequality is ensured by the feasibility of covert communication.
  If we assume that
\begin{equation}\label{as}
\theta_n = n^{-\tau}, \, \ \ \ \, 0 < \tau < 1,
\end{equation}
Square Root Law implies that the appropriate power level is $\tau \geq \frac{1}{2}$ for covert communication in the asymptotic regime. In that case, K-L distance, as a metric of discrimination with respect to the background noise at the adversary will be bounded as $n\rightarrow \infty$. Consequently, the asymptotic capacity ($ \frac{1}{2}n\log(1+n^{-\tau})$) with per channel use is zero.

A number of works focused on improving the communication efficiency by various means, such as using channel uncertainty in \cite{Seonwoo}\cite{Biao He}\cite{Khurram Shahzad}, using jammers in \cite{Sobers}\cite{Tamara V.Sobers} and other methods in \cite{B.A.Bash2}\cite{Soltani}. These methods are discussed in the asymptotic regime.
However, in practical communication, we are more concerned about the behaviors in finite blocklength regime. For example, given a finite block length $n$, how many information bits can be transmitted with a given covert criterion and maximal probability of error $\epsilon$, under which the adversary is not able to determine whether or not the transmitter is communicating effectively. When the channels are discrete memoryless, this question has been answered by Bloch's works \cite{M.Tahmasbi}\cite{M.Tahmasbi2}, where the exact second-order asymptotics of the maximal number of reliable and covert bits are characterized when the discrimination metrics are relative entropy, total variation distance (TVD) and missed detection probability, respectively. In \cite{Yu0} and \cite{Yu1}, one-shot achievability and converse bounds of Gaussian random coding under maximal power constraint are presented, and also the TVD at the adversary is roughly estimated using Pinsker's inequality. There are several reasons for us to adopt Gaussian random codes. First, Gaussian distribution is optimal in both maximizing the mutual information between the input and output ends of the legitimate receiver over AWGN channels in the asymptotic regime and minimizing KL divergence between the output and the background noise at the adversary (Theorem 5 in \cite{Ligong Wang}). It has found applications in secure chaotic spread spectrum communication systems \cite{Alan1}\cite{Alan2}. Second, TVD at the adversary is relatively easy to analyze when the codewords are Gaussian generated (or nearly Gaussian generated) than a determined codebook. In addition, random coding approach can offer us means to attain even greater achievability bounds on the number of decodable codewords. In most of previous works, K-L distance is adopted as the discrimination metric in asymptotic situation because K-L distance is convenient to analyze and compute compared with TVD. However, it is TVD that is directly related to the optimal hypothesis test. Moreover, it does not increase with the blocklength and has range $[0,1]$, hence is a normalized metric of discrimination for two probability measures. TVD is not easily obtained in general settings, and yet its close form is attainable under the assumption of Gaussian input distribution over AWGN channels, which makes it possible for us to investigate it directly with varying block length. Though our previous results provide some characterization of covert communication over AWGN channels, a thorough understanding of it requires further investigation; On one hand, an accurate characterization of the throughput in the finite blocklength regime highly depends on the accurate value of TVD at the adversary instead of its bounds. On the other hand, the direct relationship between the throughput and covert constraint is not established. Moreover, what will happen on the covertness at the adversary when $\tau$ varies in $(0, 1)$ is not fully known.

 In the current work, we will precede with our previous results on Gaussian random coding to characterize the first and second order asymptotics of the throughput. This result will establish the direct relationshp between the covert constraint and the throughput. Moreover, the TVD at the adversary is directly evaluated. Based on that, we further consider the problem in the opposite direction: give an finite block length $n$ and $\bm{snr}$ in scaling law of $n^{-\tau}$ with different $\tau \in (0, 1)$ at the main channel, how much discrimination will it give rise to at the adversary with respect to the background noise and what is its tendency when $n$ goes to infinity.

 To the best of our knowledge, our work for the first time in literature offers a comprehensive investigation about both finite block length behaviors of the throughput and the adversary's TVD. More specifically, the contributions of our work are listed as follows:\footnote{Part of this paper was submitted to ??}
\begin{itemize}

\item With given TVD upper bound $\delta$, we derive sufficient condition and necessary condition for the sending power level. As the counterpart of \cite{M.Tahmasbi2}, the first and second order of asymptotics are shown to be $O(n^{\frac{1}{2}})$ and $O(n^{\frac{1}{4}})$, respectively.
\item Under moderate blocklength assumption, analytic formula of the TVD at the adversary is obtained. From the analytic formula, we show that there is close connection between SRL and the asymptotic behavior of incomplete gamma functions.
\item The analytic formula leads to more accurate approximation of TVD, and further leads to more accurate evaluation of both achievability and converse bounds of the primary throughput, which will be illustrated by numerical results.
\item For the analytical expression, we present its series expansions with different $\bm{snr}$ for convenient evaluation. Numerical results show that they approximate the total variation distance accurately, i.e., we can provide a simple but accurate numerical description of TVD as the discrimination metric at the adversary in covert communication with properly moderate blocklength.
\item When $\tau < \frac{1}{2}$, the convergence rate that TVD at the adversary approaches to $1$ as $n\rightarrow\infty$, is proved to be $O(e^{-\frac{1}{4} n^{1-2\tau}})$. When $\tau > \frac{1}{2}$, the rate that TVD goes to $0$, is proved to be between $O(n^{1-2\tau})$ and $O(n^{\frac{1}{2}(1-2\tau)})$. These convergence rates could be quite useful for not only understanding the behavior of TVD as a metric of discrimination in probability theory, but also the practical design of covert communication.
\end{itemize}
The rest of this paper is arranged as follows. In Section II, we describe the model for covert communication over AWGN channels. In Section III, the hypothesis test at the adversary is introduced. The main results are presented in Section IV, Section V and Section VI.  We then provide numerical results in Section VII. Section VIII concluded the paper.
\section{The Channel Model and the coding Scheme}\label{model}

In this section, the channel model of covert communication over AWGN channels is presented.
An $(n,2^{nR})$ code for the Gaussian covert communication channel consists of a message set $W \in \mathcal{W}=\{1,...,2^{nR}\}$, an encoder at the transmitter Alice $f_n : \mathcal{W}\rightarrow \mathbb{R}^n, w\mapsto x^n$, and a decoder at the legitimate user Bob $g_n : \mathbb{R}^n \rightarrow \mathcal{W}, y^n \rightarrow \hat{w}$. Meanwhile, a detector is at an adversary Willie $h_n : \mathbb{R}^n \rightarrow \{0,1\}, z^n \rightarrow 0/1$. The error probability of the code is defined as $P^n_e =Pr[g_n (f_n(W))\neq W]$.
\begin{figure}
\centering
\includegraphics[width=3.5in]{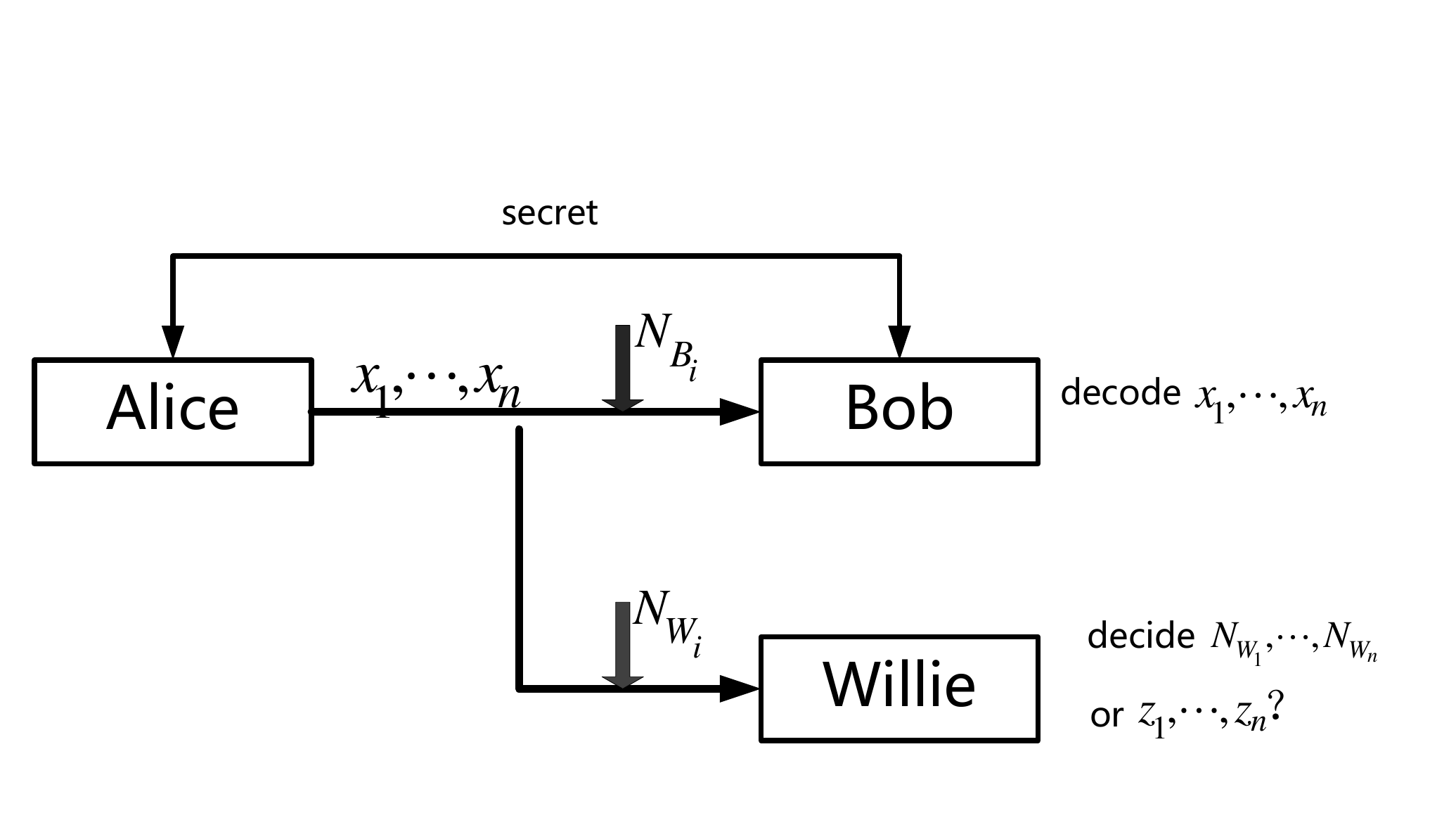}
% where an .eps filename suffix will be assumed under latex,
% and a .pdf suffix will be assumed for pdflatex; or what has been declared
% via \DeclareGraphicsExtensions.
\caption{The channel model of Gaussian LPD communication in Section \ref{model}}\label{Fig1}
\end{figure}
The channel model is defined by
\begin{eqnarray}
y_i=x_i+N_{B_i},i = 1,...,n \label{channel11}\\
z_i=x_i+N_{W_i},i = 1,...,n   \label{channel2}
\end{eqnarray}
 as shown in Fig.\ref{Fig1}, where $x^n=\{x_i\}_{i=1}^n,y^n=\{y_i\}_{i=1}^n,z^n=\{z_i\}_{i=1}^n$ denote Alice's input codeword, the legitimate user Bob's observation and the adversary Willie's observation, respectively. Variables $\{ N_{B_i},N_{W_i}, i =1, \cdots, n \}$ are independent identically distributed (i.i.d) according to Gaussian distribution $\mathcal{N}(0,\sigma_b^2)$ and $\mathcal{N}(0,\sigma^2_w)$, respectively. Each $x^n$ is a codeword from the $(n,2^{nR})$ codebook with rate $R$. The generation of the codebook will be described later. \footnote{Although the asymptotic capacity of the covert communication is zero, the rate with finite $n$ and nonzero decoding error probability $\epsilon$ could be positive.} Bob wants to decode the received vector $y^n$ with small error probability $P^n_e$. The adversary Willie tries to determine whether Alice is communicating ($h_n=1$) or not ($h_n=0$) by statistical hypothesis test, and the worst performance by Willie in detection is thus to attain the error probability of detection being $\frac{1}{2}$. Thus, Alice, who is active about her choice, is obligated to seek for a code such that $\lim_{n\rightarrow \infty}P^n_e\rightarrow 0$ and $\lim_{n\rightarrow \infty}P(h_n=0)\rightarrow \frac{1}{2}$. There is usually a secret key to assist the communication between Alice and Bob (such as the identification code for the users in spread spectrum communication), which is not the focus of this work. The interested reader may refer to \cite{Matthieu R} and \cite{Boulat A} for more details.
 For calculation convenience, it is furthter assumed that the noise levels at Alice and Willie are the same, i.e., $\sigma_b^2 = \sigma_w^2 =\sigma^2$. Each codeword is randomly selected from a subset of candidate codewords. Each coordinate of these candidates are i.i.d generated from $\mathcal{N}(0,P(n))$ where $P(n)$ is a decreasing function of $n$. The detail of selection will be discussed later. The adversary is aware that the codebook is generated from Gaussian distribution $\mathcal{N}(0, P(n))$ with blocklength $n$ but he doesn't know the specific codebook. Thus, the signal plus noise at Willie follows $\mathcal{N}(0, \sigma_1^2)$ with $\sigma_1^2 = p_n + \sigma^2$ if Alice is transmitting. As previously stated, we denote $\theta_n = \frac{p_n}{\sigma^2}$ as $\bm{snr}$, and the main concern in this work is the situation $\theta_n = n^{-\tau}$ with $\tau \in (0,1)$.

 The hypothesis test of Willie in covert communication is performed on his received signal $z^n$ which is a sample of random vector $Z^n$.
\begin{itemize}
\item
The null hypothesis $H_0$ corresponds to the situation where Alice doesn't transmit and consequently $Z^n$ has output probability distribution $\mathbb{P}_0$.
Otherwise, the received vector $Z^n$ has output probability distribution $\mathbb{P}_1$ which depends on the input distribution.
\item
The rejection of $H_0$ when it is true will lead to a false alarm with probability $\alpha$.
The acceptance of $H_0$ when it is false is considered to be a missed detection  with probability $\beta$.
\end{itemize}
The aim of Alice is to decrease the success probability of Willie's test by increasing $\alpha + \beta$, and meanwhile reliably communicating with Bob. The effect of the optimal test is usually measured by total variation distance $V_T(\mathbb{P}_1, \mathbb{P}_0)$
 which is $1-(\alpha+ \beta)$ \cite{E.Lehmann}.  The total variation distance between two probability measures $P$ and $Q$ on a sigma-algebra $\mathcal {F}$ of subsets of the sample space $\Omega$ is defined as
 \begin{equation}
 \begin{split}
  V_T(P,Q)= \underset{A \in \Omega}{\sup} \left|P(A)-Q(A)\right|.
 \end{split}
 \end{equation}
 When $V_T(\mathbb{P}_1, \mathbb{P}_0)$ is close to $0$, it is generally believed that any detector at Willie can not discriminate the induced output distribution and the distribution of noise effectively, hence can not distinguish whether or not Alice is communicating with Bob.

\section{First and Second Asymptotics}
In this section, we mainly focus on the asymptotics of covert communication over AWGN channel. Subsection A will review some results. Subsection B discusses the TVD at the adversary with the previous constructed codebook. Under that condition, we use divergence inequalities to get the sufficient and necessary condition of the power under covert constraint in Subsection B. In Subsection C, based on previous results, the achievability bound and converse bound under covert constraint are obtained, which lead to the first and second asymptotics of covert communication over AWGN channel. At first, some notions will be introduced as follows.
\begin{itemize}
\item $\mu$ is a parameter to constrain the candidates of codewords, which may depend on $n$.
\item For each $n$,  $\bar{\mathsf{F}}_n \triangleq \{x^n: \mu^2 \cdot nP\leq \|\bm{x}\|_2^2 \leq nP\}$.
\item $P$ is a decreasing function of $n$ and is usually written as $P(n)$.
\end{itemize}

\subsection{Previous Results on the Throughput}
In this section, we first present a coding scheme and corresponding achievability bound from the resulting codebook.
The generation of our codebook is descripted as follows.
  \begin{enumerate}
  \item Firstly, a set of candidates are generated. Each coordinates of these candidates is drawn from i.i.d normal distribution of variance $\mu P(n)$.
   \item Secondly, each codeword is randomly chosen from a subset $\bar{\mathsf{F}}_n$. 
   \end{enumerate}
  The decoding procedure will be sequential threshold decoding and for this codebook, we have the following normal approximation of its size.
\begin{Theorem}\label{AC1}
(Theorem 6 in \cite{Yu1}) For the AWGN channel with noise $\mathcal{N}(0, 1)$ and any $0 <\epsilon < 1$, there exists an $(n, M, \epsilon)$ code (maximal probability of error) chosen from a set $\bar{\mathsf{F}}_n$ of codewords whose coordinates are i.i.d $\sim \mathcal{N}(0, \mu P(n)), 0 < \mu < 1$ and satisfy:
 \begin{enumerate}
 \item $ \mu^2 nP(n) \leq \|\bm{x}\|_2^2 \leq nP(n)$
 \item  $\tau_0 \leq \tau_n(R) \leq \frac{n}{n+1}\epsilon$.
 \end{enumerate}
Let
\begin{equation*}
\begin{split}
&\bm{x} = [\sqrt{R}, \cdots, \sqrt{R}], \, \  \,C_{\mu}(n) = \frac{1}{2}\log(1 + \mu P(n)),\\
&\tau^{\mu}_n(R) =\frac{B_{\mu}(P,R)}{\sqrt{n}},\, \  \, B_{\mu}(P, R) = \frac{6T_{\mu}(P, R)}{\hat{V}_{\mu}(P, R)^{3/2}},\\
&T_{\mu}(P, R) = \mathbb{E}\left[|\frac{\log e}{2(1+\mu P)}[\mu P + 2\sqrt{R}Z_i -\mu PZ_i^2]|^3\right],\\
&\hat{V}_{\mu}(P, R) = \left(\frac{\log e}{2(1+P)}\right)^2(4R+ 2P^2) = V(n) \cdot \left(\frac{2R+ P^2}{2P+ P^2}\right),
\end{split}
\end{equation*}
Then we have (maximal probability of error)
\begin{equation}\label{Achi3}
\begin{split}
\log M^*_m(n, \epsilon, P(n)) \geq   &\underset{0 < \tau_0 < \epsilon}{\sup} \{ nC_{\mu}(n) +  \frac{n(R^* -\mu P(n))\log e}{2(1 + \mu P(n))}  \\
+ \sqrt{n\hat{V}_{\mu}(P(n),R^*)}&Q^{-1}\left(1 - \epsilon + \frac{2B_{\mu}(P(n),R^*)}{\sqrt{n}}\right)\\
 + \log\tau_0  + \frac{1}{2}\log n - &\log \left[\frac{2\log 2}{\sqrt{2\pi \hat{V}_{\mu}(P, R^*)}}+ 4B_{\mu}(P, R^*)\right]\}.
\end{split}
\end{equation}
The quantity $R^*$ satisfies $ x^n_0 = [\sqrt{R^*}, \cdots, \sqrt{R^*}] \in \bar{\mathsf{F}}_n$  and maximizes (\ref{logM1}).
\end{Theorem}
When $n$ is sufficiently large, there exists some $\tau_0$ that $\epsilon > \tau_n(R) > \tau_0$ holds. We further have
 \begin{equation}\label{achibound1}
    \begin{split}
     \log M^*_m(n, \epsilon, P(n)) \geq & nC_{\mu}(n) -\sqrt{nV_{\mu}(n)})Q^{-1}(\epsilon) \\
     + \frac{1}{2}\log n  + &\log\tau_0 + \log P_X[\bar{\mathsf{F}}_n] + O(1).
     \end{split}
    \end{equation}
 holds for some $\tau_0$.
\newcounter{TempEqCnt3}
\setcounter{TempEqCnt3}{\value{equation}}
\setcounter{equation}{42}
\begin{figure*}[!t]
% ensure that we have normalsize text
\normalsize
% Store the current equation number.
%\setcounter{MYtempeqncnt}{7}
% Set the equation number to one less than the one
% desired for the first equation here.
% The value here will have to changed if equations
% are added or removed prior to the place these
% equations are referenced in the main text.
%\setcounter{equation}{7}
\begin{equation}\label{logM1}
\begin{split}
&nC_{\mu}(n) +  \frac{n(R^* -\mu P(n))\log e}{2(1 + \mu P(n))} + \sqrt{n\hat{V}_{\mu}(P(n),R^*)}Q^{-1}\left(1 - \epsilon + \frac{2B_{\mu}(P(n),R^*)}{\sqrt{n}}\right)\\
& + \frac{1}{2} \log n + \log\tau_0 + \log P_X[\bar{\mathsf{F}}_n] - \log \left[\frac{2\log 2}{\sqrt{2\pi \hat{V}_{\mu}(P, R^*)}}+ 4B_{\mu}(P, R^*)\right].
\end{split}
\end{equation}
% Restore the current equation number.
%\setcounter{equation}{9}
% IEEE uses as a separator
\hrulefill
% The spacer can be tweaked to stop underfull vboxes.
\vspace*{4pt}
\end{figure*}
\setcounter{equation}{43}%

The following theorem (Formula (42) in \cite{Yu1}) provides normal approximation of the converse bound for the throughput with blocklength $n$ under maximal power constraint $P(n)$.
\begin{Theorem}\label{Conv4}
For the AWGN channel with $P = P(n)$ which is a decreasing function of $n$, and $\epsilon \in (0,1)$  and maximal power constraint under a given $n$: each codeword $c_i \in X^n$ satisfies $\|c_i \|^2 \leq nP(n)$, we have (maximal probability of error)
\begin{equation}\label{Conv2}
\begin{split}
\log &M^*_m(n, \epsilon, P(n)) \\
\leq & nC(n) -\sqrt{nV(n)})Q^{-1}(\epsilon) + \frac{1}{2}\log n + O(1).
\end{split}
\end{equation}
\end{Theorem}
Note that the converse bound is irrelevant with any coding scheme.

\subsection{TVD and The Power Level under The Coding Scheme}
Recall the process of generating the codebook: each coordinate of the candidates is generated from i.i.d Gaussian distribution $\mathcal{N}(0, \mu P(n))$ and then each codeword is selected within the region where the radius is between $\sqrt{\mu^2nP(n)}$ and $\sqrt{nP(n)}$ as shown in Figure \ref{Fig2}. The distribution $\bar{P}_{X^n}$ of the codewords is a truncated Gaussian distribution whose density function is
\begin{equation}
\footnotesize
\bm{f}(\bm{x}) = \left\{
\begin{split}
&\frac{1}{\Delta}\frac{1}{(2\pi \mu P(n))^{k/2}}e^{-\frac{\|\bm{x}\|^2}{2\mu P(n)}},    \sqrt{\mu^2nP(n)} \leq \hspace{-0.04in}\|\bm{x}\|\hspace{-0.04in}\leq \hspace{-0.04in}\sqrt{nP(n)}\\
&0,    \ \  \  \  \  \  \  \  \ \  \  \  \  \  \  \  \  \  \  \  \  \  \  \  \  \ \ \  \  \  \  \  \  \  \  \,otherwise,
\end{split}
\right.
\end{equation}
where $\Delta$ is the normalized coefficient
 \begin{equation}
 \Delta = E[1_{\{\bm{x} \in \mathfrak{B}^n_0(\sqrt{nP(n)})\backslash \mathfrak{B}^n_0(\sqrt{\mu^2nP(n)})\}}].
 \end{equation}
The distribution of the candidates $P_{X^n}$ has density function
\begin{equation}
\bm{g}(\bm{x}) = \frac{1}{(2\pi \mu P(n))^{k/2}}e^{-\frac{\|\bm{x}\|^2}{2\mu P(n)}},
\end{equation}

Let $\mathbb{P}_0$ be the n-dimensional noise distribution $\mathcal{N}(\bm{0}, \bm{I}_n)$, $\mathbb{P}_1$ be the output distribution induced by the n-dimensional Gaussian distribution $\mathcal{N}(\bm{0}, \mu P(n)\bm{I}_n)$ and let $\bar{\mathbb{P}}_1$ be the output distribution of the truncated Gaussian distribution $\bar{P}_{X^n}$. From above analysis, TVD at the adversary is written as
\begin{equation}\label{TVD1}
V_T(\bar{\mathbb{P}}_1, \mathbb{P}_0)
\end{equation}
and the power level should be chosen so that $ V_T(\bar{\mathbb{P}}_1, \mathbb{P}_0) \leq \delta$.
It is difficult to get an analytic formula of (\ref{TVD1}). We use the following bounds of TVD at the adversary.
\begin{Fact}
TVD is a distance and satisfies the triangle inequality \cite{Tsybakov}:
\begin{equation}\label{TVD2}
\begin{split}
|V_T(\mathbb{P}_1, \mathbb{P}_0) - V_T(\bar{\mathbb{P}}_1, \mathbb{P}_1)|&\leq V_T(\bar{\mathbb{P}}_1, \mathbb{P}_0) \\ &\leq V_T(\mathbb{P}_1, \mathbb{P}_0) + V_T(\bar{\mathbb{P}}_1, \mathbb{P}_1).
\end{split}
\end{equation}
\end{Fact}

\begin{Fact}
\begin{equation}
V_T(\bar{\mathbb{P}}_1, \mathbb{P}_1)\leq V_T(\bar{P}_{X^n}, P_{X^n}).
\end{equation}
\end{Fact}
Under the conditions of $\mu \in [0.7, 0.85]$ and $n \geq 400$, $V_T(\bar{P}_{X^n}, P_{X^n})$ will be small for most applications and the effect of truncation is regarded to be negligible due to sphere hardening effect. In the following analysis, it is assumed $\mu$ and $n$ are chosen as stated so that the effect of truncation is constrained under a small threshold. Without loss of generality, it is assumed that the TVD constraint is satisfied if $V_T(\mathbb{P}_1, \mathbb{P}_0) \leq \delta$ so that we can focus on the effect of the power on the asymptotics of the throughput.
\begin{figure}
\centering
\includegraphics[width=3.5in]{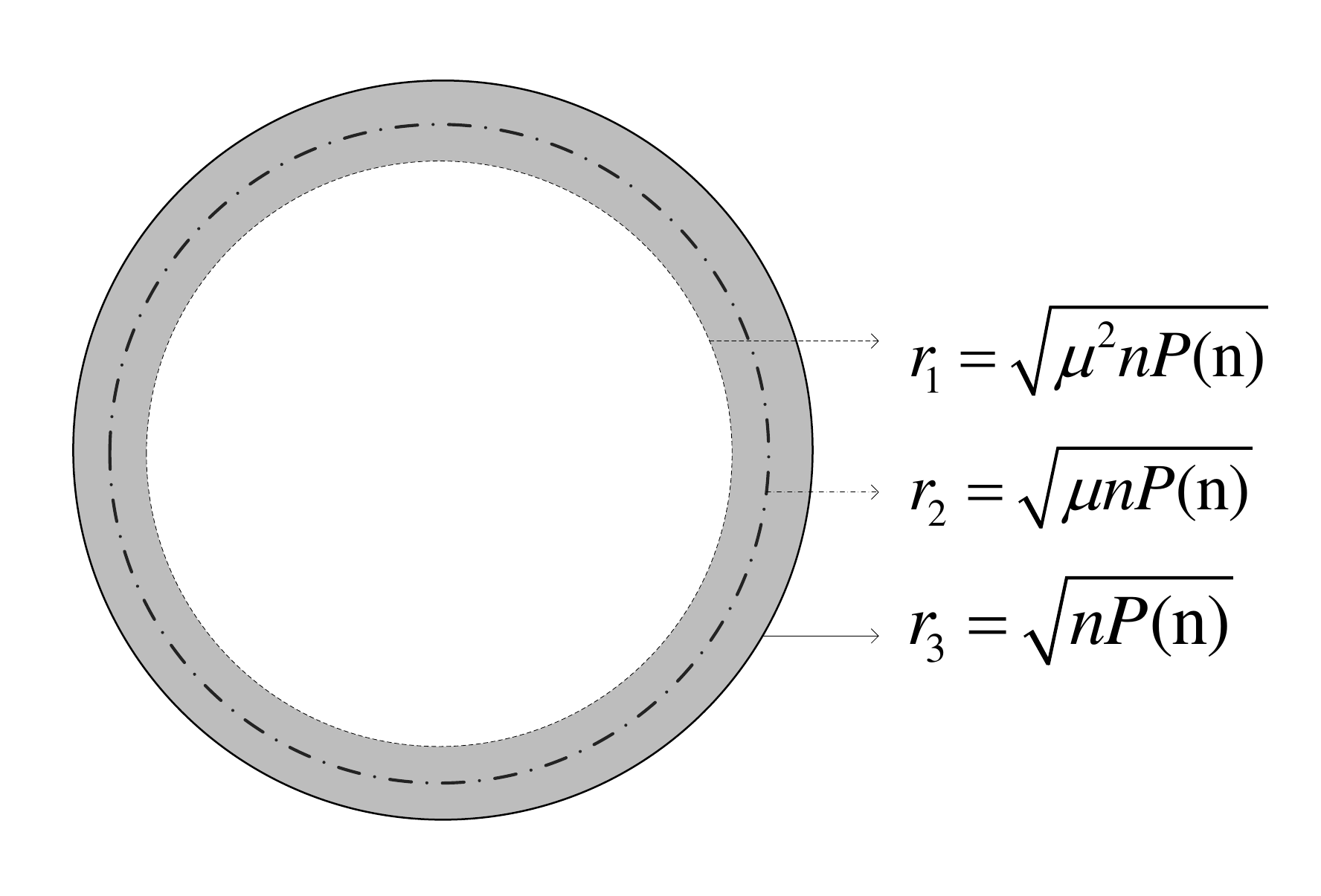}
% where an .eps filename suffix will be assumed under latex,
% and a .pdf suffix will be assumed for pdflatex; or what has been declared
% via \DeclareGraphicsExtensions.
\caption{The candidates of codewords lies in a subset of n-dimensional sphere: $\mathsf{F}_n \triangleq \{x^n: \mu^2 \cdot nP(n)\leq \|\bm{x}\|_2^2 \leq nP(n) \}$.}\label{Fig22}
\end{figure}

\subsection{Power Constraint and Divergence Inequalities}

First, we introduce some well known bounds for the total variation distance.
\begin{enumerate}
\item K-L distance bound. K-L distance is used as an upper bound of the total variation distance by Pinsker's inequality,
 \begin{equation}
 V_T(P,Q)  \leq \sqrt{\frac{1}{2}D(P\|Q)}.
 \end{equation}
 Since K-L distance is asymmetric, $D(P,Q)$ and $D(Q,P)$ are different, and both are upper bounds of the total variation distance $V_T(P,Q)$. In our case, these two K-L distances are expressed as
 \begin{eqnarray}\label{K-L distance}
 D(\mathbb{P}_1,\mathbb{P}_0) = \frac{n}{2}\left[\theta_n - \ln (1 + \theta_n)\right]\log e.\\
  D(\mathbb{P}_0,\mathbb{P}_1) = \frac{n}{2}\left[\ln(1 + \theta_n) + \frac{1}{1+ \theta_n} -1 \right]\log e
 \end{eqnarray}
 In \cite{Yan2}, it is proved that the latter is smaller than the first, which is always used as a constraint for covert communication in the form of $D(\mathbb{P}_1,\mathbb{P}_0) \leq \delta$ under the premise of Gaussian codebooks. From now on, we denote  $\sqrt{\frac{1}{2}D(\mathbb{P}_1,\mathbb{P}_0)}$ as K-L bound.
\item Hellinger distance. For probability distributions P and Q, the square of the Hellinger distance between them is defined as,
\begin{equation}\label{Hel}
H^2(P,Q) = \frac{1}{2}\int(\sqrt{dP}- \sqrt{dQ})^2.
\end{equation}
In our case, the square of the Hellinger distance is expressed as \cite{Leandro}
\begin{equation}\label{VT}
H^2(\mathbb{P}_1\|\mathbb{P}_0) = 1 - \left(\frac{2\sigma\sigma_1}{\sigma^2 + \sigma_1^2}\right)^{\frac{n}{2}}
\end{equation}
The Hellinger distance $H(P,Q)$ and the total variation distance (or statistical distance) $V_T (P,Q)$ are related as follows
\begin{equation}\label{Hellinq}
H^2(P,Q)\leq V_T(P,Q) \leq \sqrt{2}\cdot H(P,Q).
\end{equation}
Recently, Igal Sason gave an improved bound on the Hellinger distance, see Proposition 2 in \cite{Igal},
\begin{equation}\label{improved}
1 - \sqrt{1 - V_T(P,Q)^2 }\leq H^2(P,Q).
\end{equation}
From (\ref{improved}),
\begin{equation}\label{Helling2}
V_T(P,Q) \leq \sqrt{1 - (1 - H^2(P,Q))^2}.
\end{equation}
The right side of the inequality is a sharper upper bound for the total variation distance than the upper bound in (\ref{Hellinq}) and it is also sharper than K-L bound, as shown in our numerical results section.  We denote it as Hellinger upper bound.
Thus, we have 
\begin{equation}
H^2(\mathbb{P}_1,\mathbb{P}_0)\leq V_T(\mathbb{P}_1, \mathbb{P}_0) \leq \sqrt{1 - (1 - H^2(\mathbb{P}_1,\mathbb{P}_0))^2}
\end{equation}
\end{enumerate}

As we have lower bound and upper-bound on $V_T(\mathbb{P}_1, \mathbb{P}_0)$ as $B_L = H^2(\mathbb{P}_1\|\mathbb{P}_0)$, and $B_U = \sqrt{1 - (1 - H^2(\mathbb{P}_1\|\mathbb{P}_0))^2}$, i..e. $B_L \leq V_T(\mathbb{P}_1\|\mathbb{P}_0)) \leq B_U$.

(1) Let $B_L = \delta$, from which we get a power $P_{NEC}$ (i.e. necessary condition for the power). If $P > P_{NEC}$, it is impossible to achieve $V_T(\mathbb{P}_1\|\mathbb{P}_0) \leq \delta$. However, if $P \leq P_{NEC}$, we don't necessarily achieve $V_T(\mathbb{P}_1\|\mathbb{P}_0))\leq \delta $, the corresponding throughput is $\log M^*_{NEC}$. It is unlikely to attach a larger rate than this one given our TVD constraint. From (\ref{VT}) we get
 \begin{equation}
 \begin{split}
 &H^2(\mathbb{P}_1\|\mathbb{P}_0) = \delta \\
 \iff &\frac{4\sigma^2(\sigma^2+ p_n)}{(2\sigma^2 + p_n)^2} = (1- \delta)^{\frac{4}{n}}\\
 \iff &\frac{4(1+ \theta_n)}{4 + 4\theta_n + \theta_n^2} = (1-\delta)^{\frac{4}{n}}.
 \end{split}
 \end{equation}
 Denote $\eta_n = 1 + \theta_n$ and $y = \frac{1}{4}(1-\delta)^{\frac{4}{n}}$, we have
 \begin{equation}
 \frac{\eta_n}{(1 + \eta_n)^2} = y.
 \end{equation}
  Solving the above equation, we get
 \begin{equation}\label{Pnec}
 P_{NEC} = (\frac{1-2y + \sqrt{1-4y}}{2y} - 1) \cdot \sigma^2.
 \end{equation}

(2) Let $B_U = \delta$, from which we solve and find $P_{SUF}$, which suggests: if $P  \leq P_{SUF}$, we for sure can achieve $V_T(\mathbb{P}_1\|\mathbb{P}_0)\leq \delta$, but it might be too conservative to use such power. Thus, the corresponding maximal throughput is $\log M^*_{SUF}$, which is smaller than $\log M^*_{NEC}$. The actual power $P_n$ to meet the constraint of $V_T(\mathbb{P}_1\|\mathbb{P}_0) \leq \delta$, can be attained by setting $V_T(\mathbb{P}_1\|\mathbb{P}_0)) = \delta$, which should be between these two bounds, i.e.
$ P_{SUF} \leq P_n \leq P_{NEC}$.
 \begin{equation}
 \begin{split}
 &\sqrt{1 - (1 - H^2(\mathbb{P}_1\|\mathbb{P}_0))^2} = \delta^2\\
 \iff &\frac{4\sigma^2(\sigma^2+ p_n)}{(2\sigma^2 + p_n)^2} = (1- \delta)^{\frac{4}{n}}\\
 \iff &\frac{4(1+ \theta_n)}{4 + 4\theta_n + \theta_n^2} = (1-\delta)^{\frac{4}{n}}.
 \end{split}
 \end{equation}
 Denote $\eta_n = 1 + \theta_n$ as above and $y_0 = \frac{1}{4}(1-\delta^2)^{\frac{2}{n}}$, and we have
 \begin{equation}\label{Psuf}
 P_{SUF} = (\frac{1-2y_0 + \sqrt{1-4y_0}}{2y_0} - 1) \cdot \sigma^2.
 \end{equation}

 If the average power of the sending signal is smaller than $P_{SUF}$ , it is certain that $V_T(\mathbb{P}_1\|\mathbb{P}_0))\leq \delta$ will be satisfied. If the average power of the sending signal is larger than $P_{NEC}$, it is certain that $V_T(\mathbb{P}_1\|\mathbb{P}_0))\geq \delta$.
\subsection{First and Second order Asymptotics of the Maximal Thoughput}
 From the previous analysis, the results on the power requirement are applied in the achievability and converse bounds on the throughput over AWGN channel, then we can get the achievability and converse bounds on the throughput under convert constraint. If we assume that $n$ is sufficiently large, then the formula (\ref{achibound1}) could be used to characterized the first and second order asymptotics. More specifically,
\begin{enumerate}
\item
  If the quantity $P(n)$ in the achievability bound (\ref{achibound1}) is substituted by $P_{SUF}$ (\ref{Psuf}), then an achievability bound $\log M_{SUF}^*(n, \epsilon)$on the maximal throughput is obtained.
\item
  If the quantity $P(n)$ in the converse bound (\ref{Conv2}) is substituted by $P_{NEC}$ (\ref{Pnec}), then a converse bound $\log M_{NEC}^*(n, \epsilon)$ on the maximal throughput is obtained.
\end{enumerate}
The details are presented in the following theorem.
\begin{Theorem}\label{Theorem3}
 For covert communication over AWGN channel with average decoding error probability $\epsilon$ and total variation distance constraint $V_T(\mathbb{P}_1\|\mathbb{P}_0))\leq \delta$ at the adversary, the maximal throughput should satisfy:
 \begin{equation}\label{throughputupper2}
 \begin{split}
 &\log M_{NEC}^*(n, \epsilon,\delta) \leq  n\log \left[\frac{1-2y + \sqrt{1-4y}}{2y}\right] \\
 &-\sqrt{\frac{n\log^2 e}{2}\left[1 - \frac{1}{\left[\frac{1-2y_0 + \sqrt{1-4y_0}}{2y_0}\right]^2}\right]} Q^{-1}(\epsilon)+ O(\log n),
 \end{split}
 \end{equation}
 \begin{equation}\label{throughputlower2}
 \begin{split}
 &\log M_{SUF}^*(n, \epsilon,\delta) \geq  n\log \left[\frac{1-2y_0 + \sqrt{1-4y_0}}{2y_0}\right]\\
 &-\sqrt{\frac{n\log^2 e}{2}\left[1 - \frac{1}{\left[\frac{1-2y + \sqrt{1-4y}}{2y}\right]^2}\right]} Q^{-1}(\epsilon)+ O(\log n)
 \end{split}
 \end{equation}
 where $y = \frac{1}{4}(1-\delta)^{\frac{4}{n}}$ and $y_0 = \frac{1}{4}(1-\delta^2)^{\frac{2}{n}}$.
 Moreover, the first term of the maximal throughput under TVD constraint $\delta$: $\log M^*(n, \epsilon, \delta)$ is of $O(n^{\frac{1}{2}})$, and the second term is of $O(n^{\frac{1}{4}})$.
 \end{Theorem}
 \begin{proof}
 The necessary and sufficient condition on the maximal throughput are applications of Theorem \ref{Conv4} and Theorem \ref{AC1} on the power $P_{NEC}$ and $P_{SUF}$. A sufficient condition on power level will both satisfy the covert constraint and obtain the achievability bound $\log M_{SUF}^*(n, \epsilon)$. A necessary condition on power level will lead to the converse bound $\log M_{NEC}^*(n, \epsilon)$. They provide achievability and converse bounds on the maximal throughput with given TVD constraint $V_T(\mathbb{P}_1\|\mathbb{P}_0))\leq \delta$. In the following, we will analyze the order of the first and the second terms of the quantities $\log M_{NEC}^*(n, \epsilon)$ and $\log M_{SUF}^*(n, \epsilon)$. The following quantities are significant for our analysis.
\begin{itemize}
 \item The quantity $\frac{1-2y + \sqrt{1-4y}}{2y}$.

 Denote $\lambda = \sqrt{1 - 4y}$, then $\lambda = \left[1 -(1 - \delta)^{\frac{4}{n}}\right]^{\frac{1}{2}}$. If we further denote $t = - \ln (1 -\delta)^{\frac{4}{n}}$, we have $(1-\delta)^{\frac{4}{n}} = e^{-t}$ and $-t = O(\frac{1}{n})$ as $n \rightarrow \infty$.
 Now
 \begin{equation}
 \begin{split}
 \lambda = &(1 - e^{-t})^{\frac{1}{2}}\\
 = & (1 - 1 + t-\frac{t^2}{2} + \frac{t^3}{6}+ \cdots )^{\frac{1}{2}}\\
 = & O(\frac{1}{\sqrt{n}}).
 \end{split}
 \end{equation}
 Since
 \begin{equation}
 \frac{1-2y + \sqrt{1-4y}}{2y}
 = 1 + \frac{2\lambda(1 + \lambda)}{1- \lambda^2},
 \end{equation}
 we have
 \begin{equation}\label{aaa}
 \begin{split}
 \log \frac{1-2y + \sqrt{1-4y}}{2y}
 = \frac{\ln \left[1 + \frac{2\lambda(1 + \lambda)}{1- \lambda^2}\right]}{\ln 2}
 = O(\frac{1}{\sqrt{n}}).
 \end{split}
 \end{equation}
\begin{equation}\label{firstterm1}
\begin{split}
 &1 - \frac{1}{\left[\frac{1-2y_0 + \sqrt{1-4y_0}}{2y_0}\right]^2}
 =1 - \left[ \frac{1-\lambda^2}{1 + 2\lambda + \lambda^2}\right]^2
 = O(\frac{1}{\sqrt{n}}).
 \end{split}
 \end{equation}
\item The quantity $\frac{1-2y_0 + \sqrt{1-4y_0}}{2y_0}$.

Denote $\lambda_1 = \sqrt{1 - 4y_0} = \sqrt{1 - (1 - \delta^2)^{\frac{2}{n}}}$ and $s = - \frac{2}{n}\ln (1 - \delta^2)$, we have $ s = O(\frac{1}{n})$ and $(1 -\delta^2)^{\frac{2}{n}} = e^{-s}$. Moreover,
\begin{equation}
\lambda_1 = O(\frac{1}{\sqrt{n}}),
\end{equation}
\begin{equation}
\frac{1-2y_0 + \sqrt{1-4y_0}}{2y_0} = 1 + \frac{2\lambda_1(1 + \lambda_1)}{1-\lambda_1^2}.
\end{equation}
Thus,
\begin{equation}\label{Prop3bound1}
\log \frac{1-2y_0 + \sqrt{1-4y_0}}{2y_0} = O(\lambda_1) = O(\frac{1}{\sqrt{n}}),
\end{equation}
and
 \begin{equation}\label{Prop3bound2}
 \begin{split}
 &1 - \frac{1}{\left[\frac{1-2y_0 + \sqrt{1-4y_0}}{2y}\right]^2}
 =  1 - \left[ \frac{1-\lambda_1^2}{1 + 2\lambda_1 + \lambda_1^2}\right]^2
 = O(\frac{1}{\sqrt{n}}).
 \end{split}
 \end{equation}
 \end{itemize}

 Now we analyze the first and the second term of $\log M_{NEC}^*(n, \epsilon)$ and $\log M_{SUF}^*(n, \epsilon)$.
 \begin{itemize}
 \item $\log M_{NEC}^*(n,\epsilon,\delta)$ in (\ref{throughputupper2}). From the above bounds, especially (\ref{Prop3bound1}) and (\ref{Prop3bound2}), the first term is of order $O(\sqrt{n})$, and the second term is of order $O(n^{\frac{1}{4}})$.

 \item $\log M_{SUF}^*(n,\epsilon,\delta)$ in (\ref{throughputlower2}). The first term is of order $O(\sqrt{n})$, and the second term is of order $O(n^{\frac{1}{4}})$.
 \end{itemize}
 The first-order asymptotics of the maximal throughput in $\log M_{NEC}^*(n, \epsilon,\delta)$ and $\log M_{SUF}^*(n, \epsilon,\delta)$ are both $O(n^{\frac{1}{2}})$. Hence, the first-order asymptotic of $\log M^*(n, \epsilon,\delta)$ must be $O(n^{\frac{1}{2}})$.
 The second-order asymptotics of the maximal throughput in $\log M_{NEC}^*(n, \epsilon,\delta)$ and $\log M_{SUF}^*(n, \epsilon)$ are both $O(n^{\frac{1}{4}})$. Hence, the second-order asymptotic of $\log M^*(n, \epsilon, \delta)$ must be $O(n^{\frac{1}{4}})$.
 \end{proof}

\section{Analysis of Total Variation Distance in Covert Communication over AWGN Channels}
In this section, we extend the analysis of throughput to the TVD at the adversary under the same assumption that the effect of selection is negligible. In other words, we assume that $\mu$ is properly chosen and the blocklength $n$ is at least moderately large ($n \geq 500$) so that we can regard that each coordinate of the codewords is subject to normal distribution $\mathcal{N}(0, \mu P(n))$. In this case, TVD at the adversary can be approximated by $V_T(\mathbb{P}_1, \mathbb{P}_0)$.

\subsection{Analytic Formula of $V_T(\mathbb{P}_1, \mathbb{P}_0)$}
Although we have gotten upper and lower bounds of the maximal throughput, the power we use is based on divergence inequalities, which will impair the accuracy of the power and hence the throughput when the interest is on the behavior with finite $n$. In this section, we will get analytic formula of $V_T(\mathbb{P}_0,\mathbb{P}_1)$ with $\bm{snr} = \theta_n$. The formula will permit us to get accurate evaluation of TVD with given power level.
\begin{Theorem}\label{analytic}
With fixed block length $n$ and Gaussian signal with power $p_n$, the total variation distance at Willie is formulated as
\begin{equation}\label{app}
V_T(\mathbb{P}_1,\mathbb{P}_0) = \frac{1}{\Gamma(n/2)}\left[\gamma(\frac{n}{2},f(\theta_n))-\gamma(\frac{n}{2},g(\theta_n))\right].
\end{equation}
In the above formula, $n$ is the blocklength, $\theta_n = \frac{p_n}{\sigma^2}$ is the $\bm{snr}$, $\Gamma(x)$ is the well known Gamma function and $\gamma(a,x)$ is the incomplete gamma function. Moreover, $f(\theta_n)= \frac{1}{2}n\left(1+ \frac{1}{\theta_n}\right)\ln(1+ \theta_n)$ and $g(\theta_n)= \frac{1}{2}n\frac{\ln(1+ \theta_n)}{\theta_n}$.
\end{Theorem}
The proof can be found in Appendix \ref{Append2}, and it can also be obtained by geometric integration methods from \cite{Harold}.
\begin{Remark}
The incomplete gamma functions
\begin{eqnarray}
\gamma(a, z)= \int_0^{z}e^{-t}t^{a-1}dt,\\
\Gamma(a, z)= \int_z^{\infty}e^{-t}t^{a-1}dt
\end{eqnarray}
are related as follows:
$$ \gamma(a, z) + \Gamma(a, z) = \Gamma(a);$$
\end{Remark}
Theorem \ref{analytic} provides an accurate quantitative measure of the discrimination respect to the noise level at the adversary, whose input variables are the block length $n$ and $\bm{snr}$. It will help us understand the discrimination of two multivariate normal distribution with the same mean vector and different covariance matrices.
There are several interesting facts about the total variation distance at the adversary from the conclusion of Theorem \ref{analytic}.
\begin{enumerate}
\item [(a)] The numerator is the difference of two incomplete gamma functions, the first variables of which are the same, i.e., half of the blocklength.
\item [(b)] The second variables lie on the left and right of $n/2$, and the difference of which is $\frac{1}{2}n \ln(1 + \theta_n)$, i.e., the capacity multiplied by the blocklength $n$.
\end{enumerate}

\subsection{Numerical Approximation for $V_T(\mathbb{P}_1, \mathbb{P}_0)$}

Since the analytic formula for the total variation distance at the adversary is involved with Gamma function and incomplete gamma functions, it is not convenient to evaluate them in general. Therefore, it is necessary to give relatively simple formulae to evaluate these gamma functions. For Gamma function, we have String formula \cite{Simon} as asymptotic approximation,
 \begin{equation}\label{Gamma0}
 \lim_{n\rightarrow \infty } \frac{n!}{e^{-n}n^n\sqrt{2\pi n}} = 1.
 \end{equation}
For the incomplete gamma functions $\gamma(a,z)$ and $\Gamma(a, z)$, we have the following expansions for approximate evaluation:
\begin{enumerate}
\item In the case of $\mathfrak{R}(a) > -1$ and $\mathfrak{R}(a) > \mathfrak{R}(z)$, if $z$ is away from the transition point $a$ (\cite{Ferreira}, Section 3),
\begin{eqnarray}\label{gammaexpan1}
\gamma(a+1,z) = e^{-z}z^{a+1}\sum_{k=0}^{\infty}c_k(a)\Phi_k(z-a),
\end{eqnarray}
where $c_k(a)$ is expressed as
\begin{equation}
c_k(a) = \sum_{j=0}^k \frac{(-a)_j}{j!}\frac{a^{k-j}}{(k-j)!}
\end{equation}
with $(-a)_j = (-a)\cdot(-a+1)\cdots (-a+j-1)$
and has recurrence
\begin{equation}
c_{k+1}(a) =\frac{1}{k+1}[kc_k(a)-ac_{k-1}(a)].
\end{equation}
In addition,
\begin{equation}\label{ck}
c_k(a) = O(a^{\lfloor\frac{k}{2}\rfloor}),\,\ \ \, |a|\rightarrow \infty.
\end{equation}
The function $\Phi_k(z-a)$ has recurrence
\begin{equation}
\Phi_k(z-a)= \frac{1}{z-a}\left[ e^{z-a} - k\Phi_{k-1}(z-a)\right]
\end{equation}
and satisfies the following equation
\begin{equation*}
\Phi_k(z-a) = \frac{k!}{(a-z)^{k+1}} -e^{z-a}\sum_{j=0}^k\frac{k!}{(k-j)!(a-z)^{j+1}}
\end{equation*}
with $e^{z-a}$ exponentially small for $\mathfrak{R}(a) > \mathfrak{R}(z)$. We also have
$$\Phi_k(z-a) = O((z-a)^{-k-1}), \,\ \ \, |z-a|\rightarrow \infty.$$

The expansion in (\ref{gammaexpan1}) is convergent, and also asymptotic for large $a-z = O(a^{1/2+\epsilon}),\, \ \, \epsilon > 0$.
\item In the case of $\mathfrak{R}(a) > -1$ and $\mathfrak{R}(a) < \mathfrak{R}(z)$, if $z$ is away from the transition point $a$ (\cite{Ferreira}, Section 4),
\begin{equation}\label{Gammaexpan2}
\Gamma(a+1,z) \sim e^{-z}z^{a+1}\sum_{k=0}^{\infty}\frac{c_k^{*}(a)}{(z-a)^{k+1}},
\end{equation}
where $c_k(a)$ is expressed as
\begin{equation}
c_k^{*}(a) = (-1)^k\sum_{j=0}^{k}k! \frac{(-a)_j}{j!}\frac{a^{k-j}}{(k-j)!}
\end{equation}
and has recurrence
\begin{equation}
c_{k+1}^{*}(a) = -k \left[c_k^{*}(a) - ac_{k-1}^{*}(a)\right].
\end{equation}
The expansion in (\ref{Gammaexpan2}) is not convergent, nevertheless, it is asymptotic for large $a-z = O(a^{1/2+\epsilon})$ with $\epsilon > 0$.

From the expressions of $c_k$ and $c_k^{*}$, we have
\begin{equation}
c_{k}^{*}(a) = (-1)^k k! c_k(a)
\end{equation}
for case (1) and case (2).
\item For large $a$ and $z$ such that $a - z = o(a^{2/3})$, if $\|Arg(z)\|< \pi$, there is asymptotic expansion
\begin{equation}\label{Gammaexpan3}
\Gamma(a+1,z) \sim e^{-a}a^{a+1}\sum_{k=0}^{\infty}c_k(a)\Phi_k(a,z)
\end{equation}
with $$c_0(a) = 1, \,\ \ \,c_1(a) = c_2(a) =0,$$
$$\Phi_0(a,z) = \sqrt{\frac{\pi}{2a}}erfc(\frac{z-a}{\sqrt{2a}}), \, \ \ \,\Phi_1(a,z) = \frac{e^{-(z-a)^2/(2a)}}{a}$$ and for $k\geq 2$,\footnote{ $erfc$ is the complementary error function, which is defined as $ erfc(x) = 1 - erf(x) = \frac{2}{\sqrt{\pi}}\int_{x}^{\infty}e^{-t^2}dt$}
\begin{equation}
c_{k+1}(a) = \frac{1}{k+1} \left[a\cdot c_{k-2}(a) - k\cdot c_k(a)\right],
\end{equation}
\begin{equation}\label{Phi3}
\begin{split}
\Phi_{k}(a,z) = & \frac{1}{a} [(k-1)\Phi_{k-2}(a,z) \\
+& \left(\frac{z-a}{a}\right)^{k-1}\cdot e^{-\frac{(z-a)^2}{2a}}].
\end{split}
\end{equation}
\end{enumerate}

\begin{Remark}\cite{Simon}
We say that, a power series expansion $\sum_{n=0}^{\infty}a_n(z-z_0)^n$ is {\bf convergent} for $|z - z_0|< r$ with some $r\geq 0$,
provided $$R_n(x) = \sum_{n= N+1}^{\infty} a_n(z-z_0)^n \rightarrow 0,$$
as $N\rightarrow \infty$ for each fixed $z$ satisfying $|z-z_0|<r$.
We say that, a function $f(z)$ has an {\bf asymptotic} series expansion of $\sum_{n=0}^{\infty}a_n(z-z_0)^n$ as $z\rightarrow z_0$, i.e.
$$f(z) \sim \sum_0^{\infty} a_n(z-z_0)^n,$$ provided
$$R_n(x) = o((z-z_0)^N),$$ as $z\rightarrow z_0$ for each fixed $N$.
Note that, in practical terms, an asymptotic expansion can be of more value than a slowly converging expansion.
\end{Remark}

We have the following theorem by utilization of the above conclusions properly, and the details could be found in the Appendix \ref{APPend3}.

\begin{Theorem}\label{approximation}
TVD at the adversary could be approximated by
\begin{equation}\label{app1}
\frac{1}{n^{\frac{1}{4}}\sqrt{\pi}\cdot2^{\frac{5}{4}}}\sum_{k=0}^{\infty}c_k(a)\left[\Phi_k(a,g(\theta)) -  \Phi_k(a,f(\theta))\right]
\end{equation}
when $\tau \geq \frac{1}{2}$
and
\begin{equation}\label{app2}
\begin{split}
1 - &e^{-f(\theta_n)+ \frac{n}{2}}\left(\frac{f(\theta_n)}{\frac{n}{2}}\right)^{\frac{n}{2}}\frac{1}{\sqrt{\pi}n^{\frac{1}{4}}}\sum_{k=0}^{\infty}\frac{ (-1)^k k! c_k(\frac{n}{2}-1)}{(f(\theta_n)+1 -\frac{n}{2})^{k+1}} \\
 + \,\ \, &  e^{-g(\theta_n)+ \frac{n}{2}}\left(\frac{g(\theta_n)}{\frac{n}{2}}\right)^{\frac{n}{2}}\frac{1}{\sqrt{\pi}n^{\frac{1}{4}}}\sum_{k=0}^{\infty}\frac{ (-1)^{k+1} k! c_k(\frac{n}{2}-1)}{(g(\theta_n)+1 -\frac{n}{2})^{k+1}}\\
 + \,\ \,&e^{-g(\theta_n)+ \frac{n}{2}}\left(\frac{g(\theta_n)}{\frac{n}{2}}\right)^{\frac{n}{2}}\frac{1}{\sqrt{\pi}n^{\frac{1}{4}}}\sum_{k=0}^{\infty}c_k(\frac{n}{2}-1) e^{g(\theta_n)+1 -\frac{n}{2}}\\
 &\cdot\sum_{j=0}^k\frac{(-1)^j k!}{(k-j)!(g(\theta_n)+1 -\frac{n}{2})^{j+1}}
 \end{split}
\end{equation}
when $\tau < \frac{1}{2}$, respectively.

\end{Theorem}

 Though we can get some bounds and second order asymptotic on the maximal throughput of covert communication over AWGN channels by some bounds on TVD, they are usually rather rough in the finite blocklength regime. From the equations (\ref{app1}) and (\ref{app2}), the approximation of the total variation distance when $\tau \geq \frac{1}{2}$ and $\tau < \frac{1}{2}$ can be obtained. They are easy to evaluate, and numerical results show that they are good approximations for the total variation distance. From the evaluations of TVD with given values of the power level, we can approximate the proper power with given TVD constraint directly, which will lead to more accurate evaluation of the maximal throughput with different TVD constraint. Hence, Theorem \ref{approximation} provides us a tool for this approach and its importance will be more clear in Section V.

\subsection{Analysis of the Convergence Rate of $V_T(\mathbb{P}_0,\mathbb{P}_1)$ with respect to $n$}

Although the approximation numerical formulae for TVD are derived in the last section, we also wish to get its convergence rates when $n \rightarrow\infty$, which seems difficult to get from these expansions.
In the follow-on analysis, we will discuss the rates by the lower and upper bounds of $V_T(\mathbb{P}_0,\mathbb{P}_1)$ when $\tau > \frac{1}{2}$ and $\tau < \frac{1}{2}$, respectively.

The following lemma is from the definition of Hellinger distance (\ref{Hel}).
\begin{Lemma}
When $p_n \sim n ^{-\tau}\cdot \sigma^2 $ with $0 < \tau < \frac{1}{2}$, the square of the Hellinger distance $H^2(\mathbb{P}_0,\mathbb{P}_1)$ will approach to 1 when $n\rightarrow \infty$.
\end{Lemma}
\begin{proof}
For our case, the distributions $\mathbb{P}_0$ and $\mathbb{P}_1$ follow from multivariate normal distributions $N(0,\Sigma)$ and $N(0,\Sigma_1)$ with $\Sigma = \sigma^2 \cdot\mathbf{I}_n$ and $\Sigma_1 = (\sigma^2 + p_n)\cdot \mathbf{I}_n$, respectively.
  The square of the Hellinger distance of $\mathbb{P}_0$ and $\mathbb{P}_1$ is expressed as
   \begin{equation}\label{sqHell}
  H^2(\mathbb{P}_0,\mathbb{P}_1) = 1 - (\frac{2\sigma\sigma_1}{\sigma ^2 + \sigma_1^2})^{\frac{n}{2}}
   \end{equation}
   where $\sigma_1^2 = \sigma^2 + p_n$.
From the formula (\ref{sqHell}), we just need to prove that $(\frac{2\sigma\sigma_1}{\sigma ^2 + \sigma_1^2})^{\frac{n}{2}}$ approaches $0$ when the conditions are satisfied, denote $\theta = \frac{p_n}{\sigma^2} = c\cdot n^{-\tau }$ with $c$ as a constant, and the logarithm of $(\frac{2\sigma\sigma_1}{\sigma ^2 + \sigma_1^2})^{\frac{n}{2}}$  can then be formulated as follows,
\begin{equation}\label{sqHell2}
\begin{split}
&\frac{1}{2}n\ln\frac{2\sigma\sigma_1}{\sigma^2 + \sigma_1^2}\\
=&\frac{1}{2}n\ln\frac{2(1+ \theta)^{\frac{1}{2}}\sigma^2}{(2+\theta)\sigma^2}\\
=& \frac{1}{2}n\left[\ln 2(1+\theta)^{\frac{1}{2}}- \ln(2+\theta)\right]\\
=& \frac{1}{4}n\ln \frac{4+ 4cn^{-\tau }}{c^2n^{-2\tau } + 4cn^{-\tau } + 4}\\
\sim & - \frac{1}{4}n \cdot\frac{c^2n^{-2\tau }}{4cn^{-\tau } + 4}.
\end{split}
\end{equation}
When $\tau  < \frac{1}{2}$, $1 - 2\cdot\tau  > 0$ and the above logarithm will approach $-\infty$ as $n \rightarrow \infty$. Consequently, $(\frac{2\sigma\sigma_1}{\sigma ^2 + \sigma_1^2})^{\frac{n}{2}}$ approaches $0$ as $n \rightarrow \infty$ and the conclusion is obtained.
\end{proof}
\begin{Proposition}
The total variation distance between $\mathbb{P}_0$ and $\mathbb{P}_1$ will approach $1$ at the rate of  $$O(e^{-\frac{1}{4} n^{1-2\tau }})$$ when $0 < \tau  < \frac{1}{2}$ and $n \rightarrow \infty$
\end{Proposition}
\begin{proof}
If we denote $(\frac{2\sigma\sigma_1}{\sigma ^2 + \sigma_1^2})^{\frac{n}{2}}$ as $t$ in (\ref{sqHell}), then $H^2(\mathbb{P}_0,\mathbb{P}_1) = 1-t$, and $\ln t = - \frac{1}{4}n \cdot\frac{c^2n^{-2\tau }}{4cn^{-\tau } + 4} \sim -\frac{1}{4}n^{1-2\tau }$.
Thus, we have
\begin{equation*}
t \sim e^{-\frac{1}{4}n^{1-2\tau }}
\end{equation*}
 When $\tau  < \frac{1}{2}$, we have $-\frac{1}{4}n^{1-2\tau }\rightarrow -\infty$ and $ e^{-\frac{1}{4}n^{1-2\tau }} \rightarrow 0 $ as $n\rightarrow \infty$, hence the rate that  $H^2(\mathbb{P}_0,\mathbb{P}_1)$ approaches $1$ is $e^{-\frac{1}{4}n^{1-2\tau }}$. Furthermore,
\begin{equation*}
H(\mathbb{P}_0,\mathbb{P}_1) = \sqrt{1-t} = 1 -\frac{1}{2}t + o(t) \, \ \ \, as \, \ \ \,t \rightarrow 0.
\end{equation*}
Therefore, $H(\mathbb{P}_0,\mathbb{P}_1)$ approaches $1$ at the same rate.
Consequently, from (\ref{Hellinq}), the rate that $V_T(\mathbb{P}_0,\mathbb{P}_1)$ approaches $1$ when $\tau  < \frac{1}{2}$ is $c\cdot e^{-\frac{1}{4} n^{1-2\tau }}$, where $c$ is a constant.
\end{proof}

Next, we consider the situation where $\tau  > \frac{1}{2}$.
\begin{Proposition}\label{Prop}
The total variation distance between $\mathbb{P}_0$ and $\mathbb{P}_1$ will approach $0$ at the rate between $O(n^{1-2\tau })$ and $O(n^{\frac{1}{2}(1-2\tau )})$ if $p_n \sim n ^{-\tau }\cdot \sigma^2 $ with $\tau  > \frac{1}{2}$ and $n \rightarrow \infty$.
\end{Proposition}
\begin{proof}
From (\ref{sqHell2}), the logarithm of $(\frac{2\sigma\sigma_1}{\sigma ^2 + \sigma_1^2})^{\frac{n}{2}}$ will approach $0$ from the left as $n\rightarrow \infty$, hence $(\frac{2\sigma\sigma_1}{\sigma ^2 + \sigma_1^2})^{\frac{n}{2}}$ will approach $1$ from the left. Therefore, $H^2(\mathbb{P}_0,\mathbb{P}_1)$ will approach $0$. From (\ref{Hellinq}), $V_T(\mathbb{P}_0,\mathbb{P}_1)$ will approach $0$. When $\tau  > \frac{1}{2}$, $-\frac{1}{4}n^{1-2\tau }$ will approach $0$ from the negative axis. From Taylor expansion, $e^x = 1 + x + o(x)$, we have
\begin{equation}
t \sim e^{-\frac{1}{4}n^{1-2\tau }}= 1 - \frac{1}{4}n^{1-2\tau } + o( \frac{1}{4}n^{1-2\tau }).
\end{equation}
The rate that $e^x$ approaches $1$ is almost determined by the rate that $x$ goes $0$. Therefore, $t$ approaches $1$ at the rate of $\frac{1}{4}n^{1-2\tau }$, i.e., $H^2(P,Q)$ approaches $0$ at the rate of $\frac{1}{4}n^{1-2\tau }$ when $\tau  > \frac{1}{2}$.
\begin{equation*}
H(P,Q) = \sqrt{1-t} \sim  \sqrt{\frac{1}{4}n^{1-2\tau } + o( \frac{1}{4}n^{1-2\tau })}\sim \frac{1}{2}n^{\frac{1}{2}(1-2\tau )}
\end{equation*}
Thus, the rate that $V_T(\mathbb{P}_0,\mathbb{P}_1)$ approaches $0$ is between $O(n^{1-2\tau })$ and $O(n^{\frac{1}{2}(1-2\tau )})$.
\end{proof}
The convergence rates of TVD provide a lot of information for covert communication over AWGN channels in finite blocklength regime, which are listed as follows.
\begin{Remarks}
\begin{itemize}
\item With given $\epsilon > 0$, we can only talk about finite blocklength $n$. The blocklength $n$ and the power level ($\tau $) should be chosen carefully to satisfy bounds on given decoding error probability $\epsilon$ and TVD $\delta$.
\item  Under any given $ 0 <\delta <1$, and a fixed $\tau > \frac{1}{2}$, as $n$ increases it will definitely satisfy the requirement on the upper-bound imposed on TVD.
\item If $\tau  < 1/2$, increasing $n$ will eventually violate any given upper bound $0<\delta<1$ on TVD.
\item With given $\delta$, if $p_n = C\cdot n^{-\tau }$ with proper constant $C$ and $\tau  = 1/2$, we can increase $n$ to satisfy any small decoding error probability $\epsilon$ without worrying about the violation of TVD bound $\delta$ since the total variation distance will be stationary. Moreover we can also provide the second order asymptotics in this case for $\log(M_n)$.

\item The rate can be also testified by using (\ref{improved}). In our case, we have
\begin{equation}
\begin{split}
&\sqrt{1 - (1 - H^2(\mathbb{P}_0,\mathbb{P}_1))^2} \\
= &\sqrt{1 - t^2}\\
\sim &\sqrt{1-e^{-\frac{1}{2}n^{1-2\tau }}}\\
\sim &\sqrt{1-\left[1-\frac{1}{2}n^{1-2\tau }+ \frac{1}{8}n^{2-4\tau } + \cdots\right]}\\
\sim & \frac{\sqrt{2}}{2}n^{\frac{1}{2}(1-2\tau )}.
\end{split}
\end{equation}
Hence, we have the same rate upper bound as Proposition 2.
\item The rate bound in the last proposition can also be testified from the bound of total variation distance in terms of K-L distance. From (22) in \cite{Igal},
\begin{equation}
D(P,Q)\geq \log \left(\frac{1}{1-V_T(P,Q)^2}\right).
\end{equation}
We have
\begin{equation}\label{K-L bound2}
V_T(P,Q)\leq \sqrt{1- e^{-D(P,Q)}}.
\end{equation}
From (34) in \cite{Igal},
\begin{equation}\label{K-L lower bound}
V_T(P,Q)\geq \left(\frac{1-\beta}{\log\frac{1}{\beta}}\right)D(P,Q).
\end{equation}
The K-L distance in our case can be reformulated as follows
\begin{equation}\label{K-L upp}
\begin{split}
  D(\mathbb{P}_0,\mathbb{P}_1) = &\frac{n}{2}\left[\ln(1 + \theta_n) + \frac{1}{1+ \theta_n} -1 \right]\log e\\
= & \frac{n}{2\ln2}\left[\theta_n -\frac{1}{2}\theta_n^2 + 1 -\theta_n + \theta_n^2 -1 + o(\theta_n^2)\right]\\
= & \frac{n}{2\ln2}\left[ \frac{1}{2}\theta_n^2 + o(\theta_n^2)\right]\\
\sim & \frac{1}{4\ln2}n^{1-2\tau }.
\end{split}
\end{equation}
When $\tau  > \frac{1}{2}$, it goes to $0$ at rate $O(n^{1-2\tau })$. Hence, from (\ref{K-L lower bound}), the lower bound goes to $0$ at the rate of $O(n^{1-2\tau })$. For the upper bound, from (\ref{K-L bound2}),
\begin{equation}
\begin{split}
 \sqrt{1- e^{-D(\mathbb{P}_1,\mathbb{P}_0)}} = &\sqrt{1 - \left[1 - D(\mathbb{P}_1,\mathbb{P}_0)+ o(n^{1-2\tau })\right]}\\
 \sim & \sqrt{\frac{1}{4\ln2}n^{1-2\tau }+o(n^{1-2\tau })}.
\end{split}
\end{equation}
Hence, the upper bound of the total variation distance goes to $0$ at the rate of $O(n^{\frac{1}{2}(1-2\tau )})$.
In summary, we also get that the rate that the total variation distance goes to $0$ is between $O(n^{1-2\tau })$ and $O(n^{\frac{1}{2}(1-2\tau )})$.
\end{itemize}
\end{Remarks}

\section{Numerical Results}\label{NUMP}
In this section, the numercial results are presented. The main results in Section III and Section IV are testified. Since we can only limit the effect of truncation by choosing proper $\mu$ and moderately large blocklength. In the following, the least blocklength is $500$ when $\delta = 0.01$, then the effect of selection (or truncation) is negligible. The least blocklength is even larger when $\delta = 0.01$. In these circumstances, the codewords could be regarded as Gaussian generated and the $TVD$ at the adversary is approximated as $V_T(\mathbb{P}_0,\mathbb{P}_1)$.

In Fig.\ref{Fig2} and Fig.\ref{Fig3}, the necessary condition $P_{NEC}$ and sufficient condition $P_{SUF}$ of the power for $V_T(\mathbb{P}_1\|\mathbb{P}_0) \leq \delta$ at different blocklength $n$ are plotted when $\delta$ is fixed as $0.1$ and $0.01$, respectively. They are compared with the power approximated directly from formula (\ref{app}). We can see that the sufficient condition of the power for covert constraint is quite close to the approximation when $\delta = 0.1$ or $\delta = 0.01$. The maximal value of power proper for covert constraint $V_T(\mathbb{P}_1\|\mathbb{P}_0) \leq \delta$  will always be in the zone between  two curves of sufficient and necessary conditions with $0 < \delta < 1$. In Fig.\ref{Fig6}, we plot the necessary condition $P_{NEC}$ and sufficient condition $P_{SUF}$ of the power for $V_T(\mathbb{P}_1\|\mathbb{P}_0) \leq \delta$ with different $\delta$  with fixed blocklength $n = 2000$. It is obvious that the approximation of power will be in the zone between the curve of $P_{SUF}$ and the curve of $P_{NEC}$ .
From the analytic solution in Proposition 2, the behavior of $V_T(\mathbb{P}_1\|\mathbb{P}_0)$ when the power scaling law follows $\bm{snr} = \theta_n = n^{-\tau}$ with $\tau < \frac{1}{2}$ at the main channel with different $\tau$ can be found in Fig.\ref{Fig7}. As $n$ tends to infinity, we can see that $V_T(\mathbb{P}_1\|\mathbb{P}_0)$ approaches $1$ exponentially if $\tau < \frac{1}{2}$, and the rate it approaches $0$ is polynomial if $\tau > \frac{1}{2}$. When $\tau = \frac{1}{2}$, $V_T(\mathbb{P}_1\|\mathbb{P}_0)$ will be stationary even when $n$ is very large.

We plot TVD $V_T(\mathbb{P}_1\|\mathbb{P}_0)$ from (\ref{app}), the square of the Hellinger distance from (\ref{sqHell}) , Hellinger upper bound from (\ref{improved}) and the approximation expansion from (\ref{app2}) when $\tau < \frac{1}{2}$ in Fig.\ref{Fig8}. It is obvious that the approximation from (\ref{app2}) is quite accurate and can be used in practical performance analysis of covert communication. Moreover, the validity that $V_T(\mathbb{P}_1\|\mathbb{P}_0)$ goes to $1$ exponentially is demonstrated again.

These quantities $V_T(\mathbb{P}_0,\mathbb{P}_1), \sqrt{\frac{1}{2}D(\mathbb{P}_0,\mathbb{P}_1)}$ (K-L bound), Hellinger upper bound (\ref{improved}) and the approximation of $V_T(\mathbb{P}_0,\mathbb{P}_1)$ in (\ref{alphageq})  with $\tau  > \frac{1}{2}$ are plotted in Fig.\ref{Fig9}. The accuracy of the approximation (\ref{app1}) is obvious. It is also clear that the rates that they approach $0$ are polynomial. Moreover, the validity of these bounds is testified and the relationship between them with finite block length is demonstrated.
\begin{figure}
\includegraphics[width=3.5in]{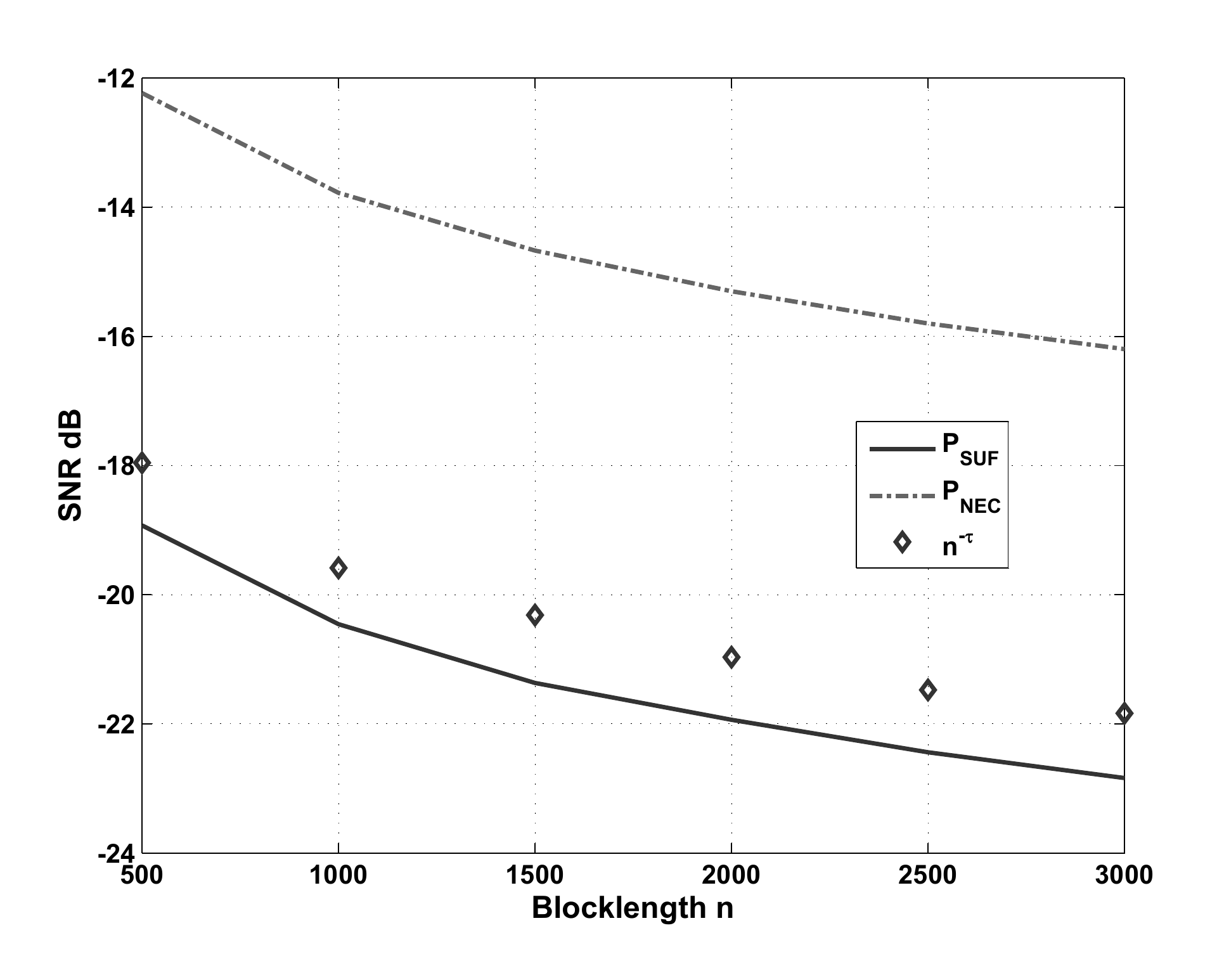}
% where an .eps filename suffix will be assumed under latex,
% and a .pdf suffix will be assumed for pdflatex; or what has been declared
% via \DeclareGraphicsExtensions.
\caption{The sufficient, necessary condition and the approximation of the power $\theta_n$ for $\delta = 0.1$.}\label{Fig2}
\end{figure}
 \begin{figure}
\includegraphics[width=3.5in]{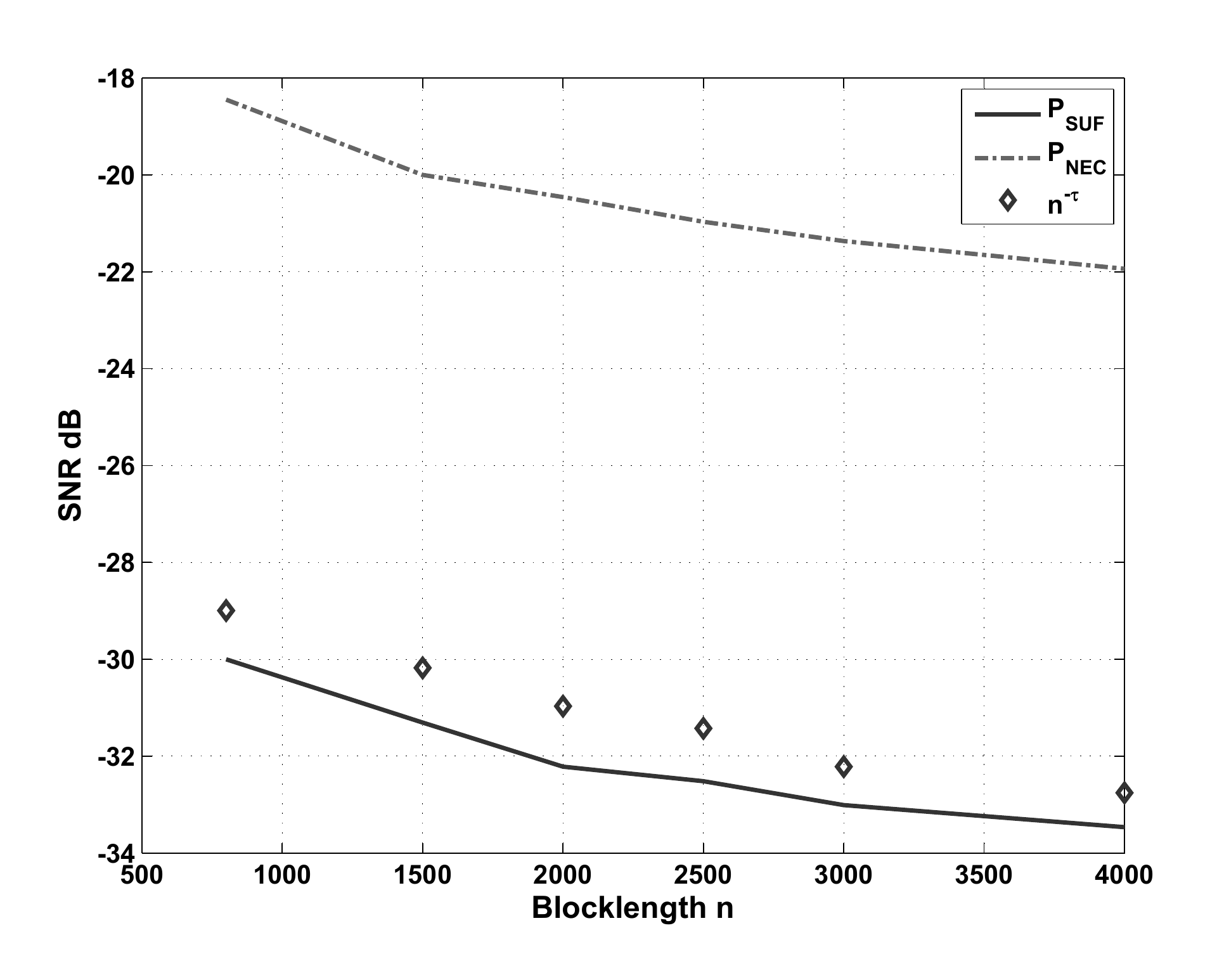}
% where an .eps filename suffix will be assumed under latex,
% and a .pdf suffix will be assumed for pdflatex; or what has been declared
% via \DeclareGraphicsExtensions.
\caption{The sufficient, necessary condition and the approximation of the power $\theta_n$ for $\delta = 0.01$.}\label{Fig3}
\end{figure}
 \begin{figure}
\includegraphics[width=3.5in]{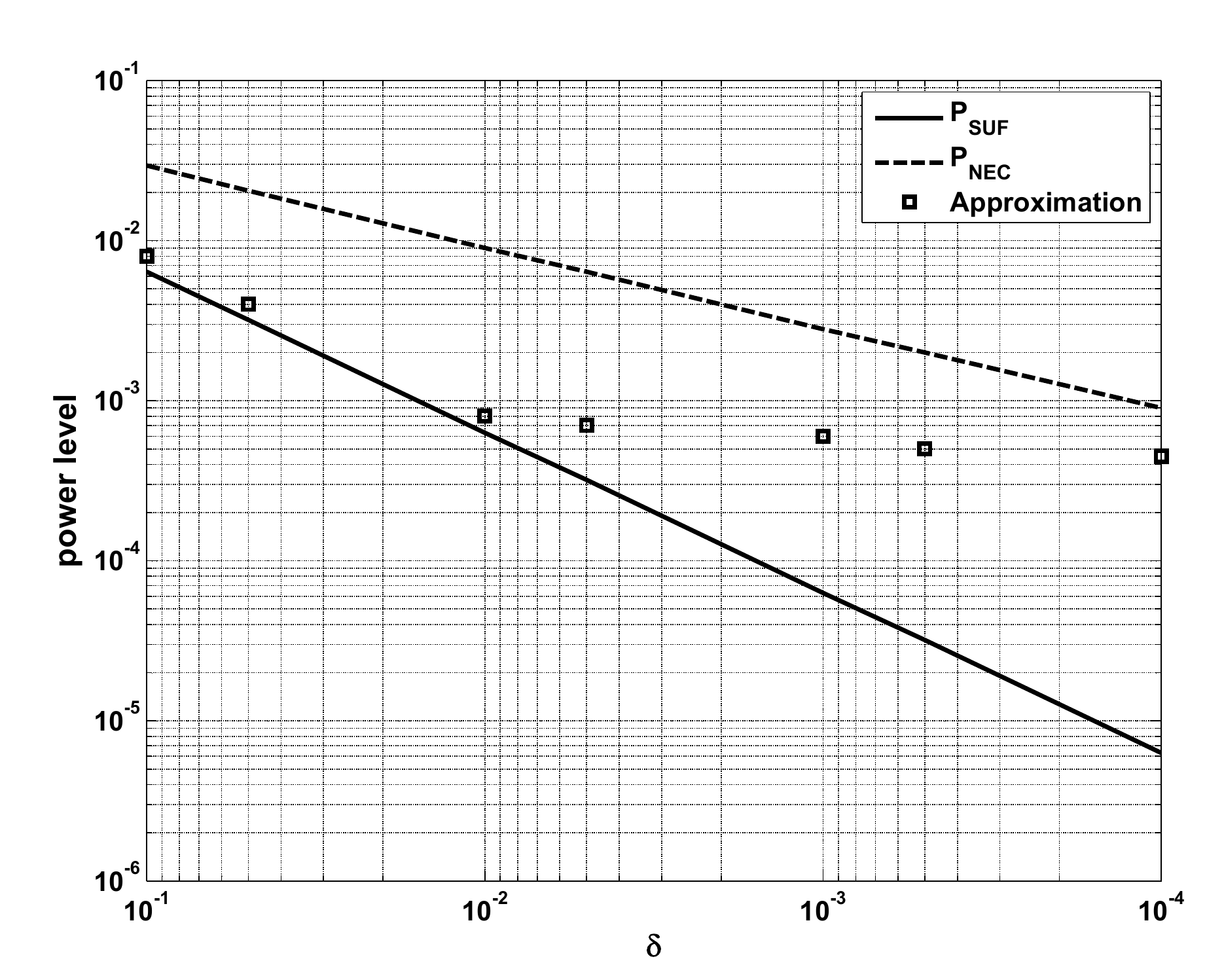}
% where an .eps filename suffix will be assumed under latex,
% and a .pdf suffix will be assumed for pdflatex; or what has been declared
% via \DeclareGraphicsExtensions.
\caption{The sufficient and necessary condition for the power and the approximation for covert communication over AWGN channel for varying $\delta$ with fixed blocklength $n = 2000$.}\label{Fig6}
\end{figure}

\begin{figure}
\includegraphics[width=3.5in]{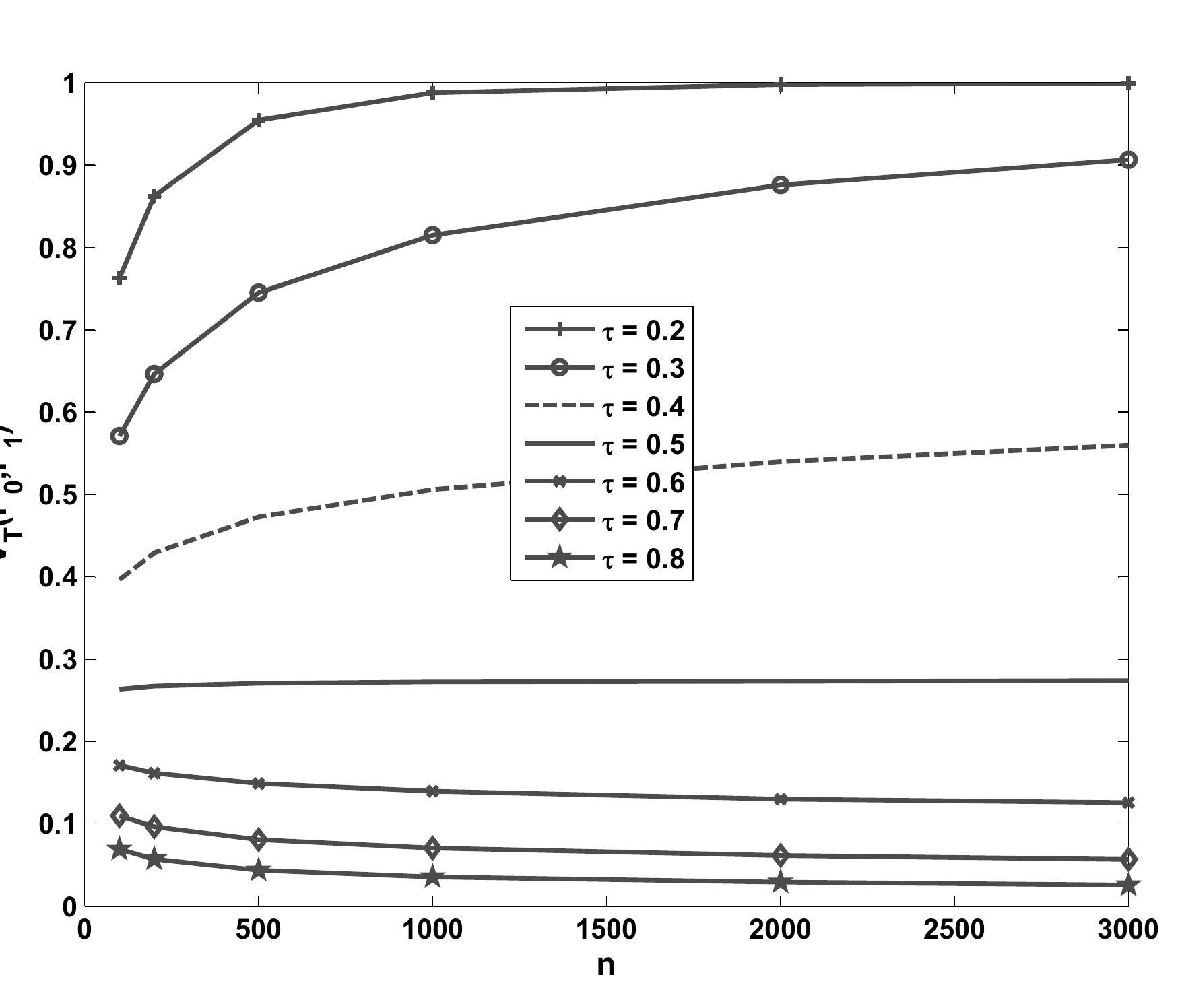}
% where an .eps filename suffix will be assumed under latex,
% and a .pdf suffix will be assumed for pdflatex; or what has been declared
% via \DeclareGraphicsExtensions.
\caption{Comparison of different total variation distances with the the code of length $n$ and different $\tau$.}\label{Fig7}
\end{figure}

\begin{figure}
\includegraphics[width=3.5in]{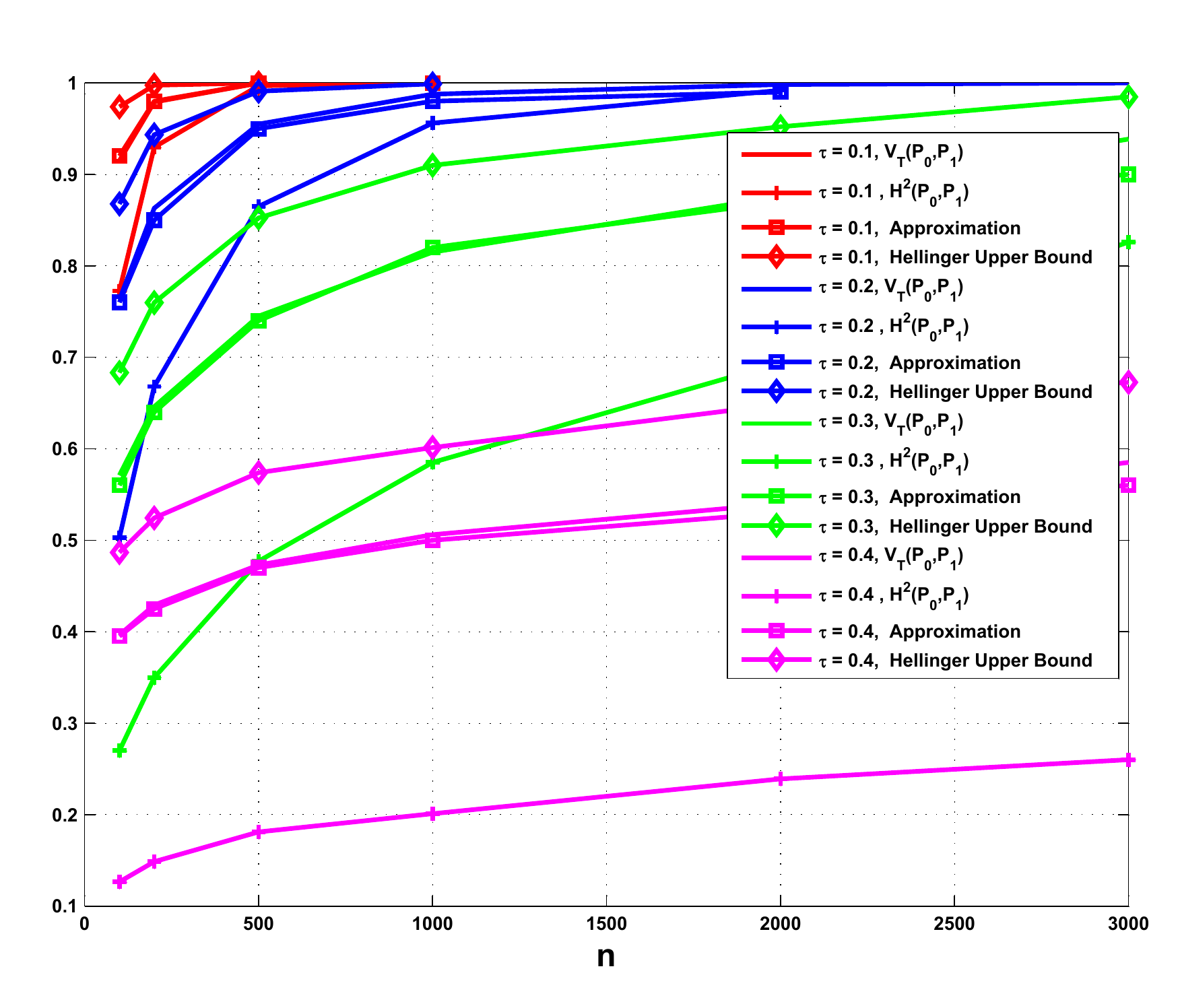}
% where an .eps filename suffix will be assumed under latex,
% and a .pdf suffix will be assumed for pdflatex; or what has been declared
% via \DeclareGraphicsExtensions.
\caption{Comparison between $V_T(\mathbb{P}_0,\mathbb{P}_1), H^2(\mathbb{P}_0,\mathbb{P}_1)$, Hellinger upper bound and the approximation by the expansions with the length of the code $n$ with $\tau < \frac{1}{2}$.}\label{Fig8}
\end{figure}

\begin{figure}
\includegraphics[width=3.5in]{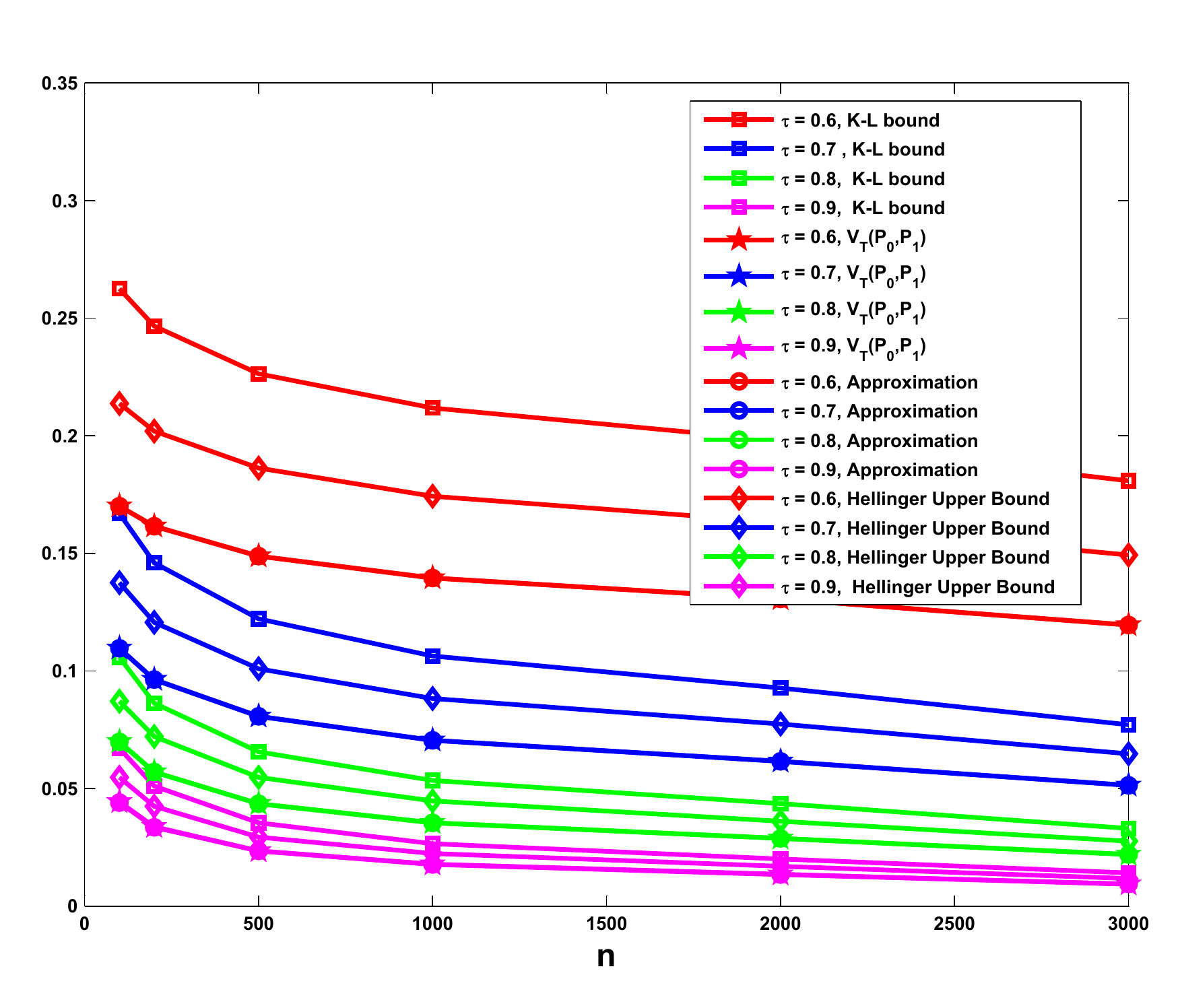}
% where an .eps filename suffix will be assumed under latex,
% and a .pdf suffix will be assumed for pdflatex; or what has been declared
% via \DeclareGraphicsExtensions.
\caption{Comparison between $V_T(\mathbb{P}_0,\mathbb{P}_1), \sqrt{\frac{1}{2}D(\mathbb{P}_0,\mathbb{P}_1)}$ (K-L bound), Hellinger upper bound and the approximation by the expansions with the length of the code $n$ with $\tau > \frac{1}{2}$. Since the square of Hellinger distance is too loose as a lower bound of the total variation distance when $\tau > \frac{1}{2}$, it is not plotted here.}\label{Fig9}
\end{figure}

\section{Conclusion}
In this work we consider covert communication over AWGN channels in finite block length regime. The maximal throughput with TVD constraint is investigated and the first and second asymptotics are obtained, which extends Square Root Law for covert communication. We also got close formula of TVD, between the distributions of the noise and the signal plus the noise at the adversary. The numerical approximation expressions for TVD with different signal noise ratio levels were further discussed, which are helpful for practical design and analysis of covert communication. Furthermore, our investigation about the convergence rates of TVD when $n\rightarrow \infty$ are meaningful for understanding the total variation distance as a metric of discrimination of Gaussian distributions with different variances. In future work we plan on investigating covert communication over MIMO systems.
\appendices
\section{Proof of Theorem \ref{analytic}} \label{Append2}

We have $1 -(\tau  + \beta) = V_T(\mathbb{P}_1,\mathbb{P}_0) = \frac{1}{2}\|\mathbb{P}_1(\mathbf{x}) - \mathbb{P}_0(\mathbf{x})\|_1$, $\mathbb{P}_0$ and $\mathbb{P}_1$ are $n$-product Gaussian distributions with zero mean and variance $\sigma^2$ and $\sigma_1^2 = \sigma^2 + p_n$, respectively. $p_n$ is the average power per symbol. We derive from (\ref{integration0}) and get (\ref{inital}) by integrating the variable in the $n$ dimension ball.
%\newcounter{MYtempeqncnt}
%\setcounter{equation}{3}
\newcounter{TempEqCnt8}
\setcounter{TempEqCnt8}{\value{equation}}
\setcounter{equation}{131}
\begin{figure*}[!t]
% ensure that we have normalsize text
\normalsize
% Store the current equation number.
% Set the equation number to one less than the one
% desired for the first equation here.
% The value here will have to changed if equations
% are added or removed prior to the place these
% equations are referenced in the main text.

\begin{equation}\label{integration0}
\begin{split}
&\|\mathbb{P}_0(\mathbf{x}) - \mathbb{P}_1(\mathbf{x})\|_1\\
= &\iint\cdots\int_{x_1,x_2,...,x_n}\left|\frac{1}{(2\pi\sigma^2)^{n/2}}e^{-\frac{\sum_{i=1}^{n}x_i^2}{2\sigma^2}} - \frac{1}{(2\pi\sigma_1^2)^{n/2}}e^{-\frac{\sum_{i=1}^{n}x_i^2}{2\sigma_1^2}}\right|dx_1\cdots dx_n\\
= &\iint\cdots\int_{x_1,x_2,...,x_n}\frac{1}{(2\pi\sigma^2)^{n/2}}e^{-\frac{\sum_{i=1}^{n}x_i^2}{2\sigma^2}} \left|1 - (\frac{\sigma^2}{\sigma_1^2})^{n/2} e^{-\frac{\sum_i x_i^2}{2}(\frac{1}{\sigma_1^2}-\frac{1}{\sigma^2})}\right|dx_1\cdots dx_n \\
\overset{(a)}{=} & \iint\cdots\int_{\sum_i x_i^2 \leq \frac{n(\sigma_1\sigma)^2\ln(1 + \frac{p_n}{\sigma^2})}{p_n}}\left(\frac{1}{(2\pi\sigma^2)^{n/2}}e^{-\frac{\sum_{i=1}^{n}x_i^2}{2\sigma^2}} - \frac{1}{(2\pi\sigma_1^2)^{n/2}}e^{-\frac{\sum_{i=1}^{n}x_i^2}{2\sigma_1^2}}\right)dx_1\cdots dx_n\\
+  & \iint\cdots\int_{\sum_i x_i^2 \geq \frac{n(\sigma_1\sigma)^2\ln(1 + \frac{p_n}{\sigma^2})}{p_n} }\left(\frac{1}{(2\pi\sigma_1^2)^{n/2}}e^{-\frac{\sum_{i=1}^{n}x_i^2}{2\sigma_1^2}} - \frac{1}{(2\pi\sigma^2)^{n/2}}e^{-\frac{\sum_{i=1}^{n}x_i^2}{2\sigma^2}}\right)dx_1\cdots dx_n\\
\overset{(b)}{=}& 2 \cdot\iint\cdots\int_{\sum_i x_i^2 \leq \frac{n(\sigma_1\sigma)^2\ln(1 + \frac{p_n}{\sigma^2})}{p_n}}\left(\frac{1}{(2\pi\sigma^2)^{n/2}}e^{-\frac{\sum_{i=1}^{n}x_i^2}{2\sigma^2}} - \frac{1}{(2\pi\sigma_1^2)^{n/2}}e^{-\frac{\sum_{i=1}^{n}x_i^2}{2\sigma_1^2}}\right)dx_1\cdots dx_n\\
\end{split}
\end{equation}
\begin{equation}\label{inital}
\begin{split}
V_T(\mathbb{P}_1,\mathbb{P}_0) = \cdot\iint\cdots\int_{\sum_i x_i^2 \leq \frac{n(\sigma_1\sigma)^2\ln(1 + \frac{p_n}{\sigma^2})}{p_n}}\left(\frac{1}{(2\pi\sigma^2)^{n/2}}e^{-\frac{\sum_{i=1}^{n}x_i^2}{2\sigma^2}} - \frac{1}{(2\pi\sigma_1^2)^{n/2}}e^{-\frac{\sum_{i=1}^{n}x_i^2}{2\sigma_1^2}}\right)dx_1\cdots dx_n.
\end{split}
\end{equation}
% Restore the current equation number.
% IEEE uses as a separator
\hrulefill
% The spacer can be tweaked to stop underfull vboxes.
\vspace*{4pt}
\end{figure*}

\setcounter{equation}{133}%
In the derivation, the equation (a) follows from the following inequalities:
\begin{equation}
\begin{split}
&\frac{1}{(2\pi\sigma^2)^{n/2}}e^{-\frac{\sum_{i=1}^{n}x_i^2}{2\sigma^2}} - \frac{1}{(2\pi\sigma_1^2)^{n/2}}e^{-\frac{\sum_{i=1}^{n}x_i^2}{2\sigma_1^2}} \geq 0\\
\iff &1 - (\frac{\sigma^2}{\sigma_1^2})^{n/2}e^{-\frac{\sum_i x_i^2}{2}(\frac{1}{\sigma_1^2}-\frac{1}{\sigma^2})}\geq 0 \\
\iff &e^{\frac{\sum_i x_i^2}{2}(\frac{1}{\sigma_1^2}-\frac{1}{\sigma^2})} \geq (\frac{\sigma^2}{\sigma_1^2})^{n/2}\\
\iff &\frac{\sum_i x_i^2}{2}(\frac{1}{\sigma_1^2}-\frac{1}{\sigma^2}) \geq \frac{n}{2}(\ln\sigma^2 -\ln\sigma_1^2)\\
\iff & \sum_i x_i^2 \leq \frac{n(\ln\sigma^2-\ln(\sigma^2 + p_n))}{\frac{1}{\sigma_1^2}-\frac{1}{\sigma^2}}\\
\iff & \sum_i x_i^2 \leq \frac{n(\sigma_1\sigma)^2\ln(1 + \frac{p_n}{\sigma^2})}{p_n}.
\end{split}
\end{equation}

The equation (b) follows from the following equalities:
\begin{equation}
\begin{split}
&\iint\cdots\int_{\sum_i x_i^2 \geq \frac{n(\sigma_1\sigma)^2\ln(1 + \frac{p_n}{\sigma^2})}{p_n}}\frac{1}{(2\pi\sigma^2)^{n/2}}e^{-\frac{\sum_{i=1}^{n}x_i^2}{2\sigma^2}} \\
=& 1 - \iint\cdots\int_{\sum_i x_i^2 \leq \frac{n(\sigma_1\sigma)^2\ln(1 + \frac{p_n}{\sigma^2})}{p_n}}\frac{1}{(2\pi\sigma^2)^{n/2}}e^{-\frac{\sum_{i=1}^{n}x_i^2}{2\sigma^2}} \\
& \iint\cdots\int_{\sum_i x_i^2 \geq \frac{n(\sigma_1\sigma)^2\ln(1 + \frac{p_n}{\sigma^2})}{p_n}}\frac{1}{(2\pi\sigma^2)^{n/2}}e^{-\frac{\sum_{i=1}^{n}x_i^2}{2\sigma^2}} \\
=& 1 - \iint\cdots\int_{\sum_i x_i^2 \leq \frac{n(\sigma_1\sigma)^2\ln(1 + \frac{p_n}{\sigma^2})}{p_n}}\frac{1}{(2\pi\sigma^2)^{n/2}}e^{-\frac{\sum_{i=1}^{n}x_i^2}{2\sigma^2}}.
\end{split}
\end{equation}

Denote $R^2= \frac{n(\sigma_1\sigma)^2\ln(1 + \frac{p_n}{\sigma^2})}{p_n}$;
to calculate the integration of (\ref{inital}), we need to calculate the following integration,
\begin{equation}\label{ing1}
 \iint\cdots\int_{\sum_i x_i^2 \leq R^2 }\frac{1}{(2\pi\sigma^2)^{n/2}}e^{-\frac{\sum_{i=1}^{n}x_i^2}{2\sigma^2}}dx_i\cdots dx_n.
 \end{equation}
 By the following variable substitution,
 \begin{align}\notag
 \begin{cases}
  x_1 = r\cos\theta_1 \\
  x_2 = r\sin\theta_1 cos\theta_2\\
  \cdots \\
  x_{n-1}  = r \sin\theta_1\sin\theta_2\sin\theta_3\cdots \cos\theta_{n-1} \\
  x_n  = r \sin\theta_1\sin\theta_2\sin\theta_3\cdots \sin\theta_{n-1} \\
 \end{cases}\\
 0\leq r \leq R, 0 < \theta_1, \theta_2, \cdots, \theta_{n-2} < \pi, 0 < \theta_{n-1}< 2\pi \notag
 \end{align}
the integration can be rewritten as (\ref{integration21}).
%\newcounter{MYtempeqncnt1}
\newcounter{TempEqCnt2}
\setcounter{TempEqCnt2}{\value{equation}}
\setcounter{equation}{136}
\begin{figure*}[!t]
% ensure that we have normalsize text
\normalsize
% Store the current equation number.
%\setcounter{MYtempeqncnt}{7}
% Set the equation number to one less than the one
% desired for the first equation here.
% The value here will have to changed if equations
% are added or removed prior to the place these
% equations are referenced in the main text.
%\setcounter{equation}{7}
\begin{equation}\label{integration21}
\begin{split}
& \, \  \,\iint\cdots\int_{\sum_i x_i^2 \leq R^2 }\frac{1}{(2\pi\sigma^2)^{n/2}}e^{-\frac{\sum_{i=1}^{n}x_i^2}{2\sigma^2}}dx_i\cdots dx_n\\
 & = \iint\cdots\int_{0 < r^2 \leq R^2, 0 < \theta_1, \theta_2, \cdots, \theta_{n-2} < \pi, 0 < \theta_{n-1}< 2\pi  }\frac{1}{(2\pi\sigma^2)^{n/2}}e^{-\frac{r^2}{2\sigma^2}} r^{n-1}\sin ^{n-2}\theta_1\sin ^{n-3}\theta_2\cdots \sin \theta_{n-2}dr d\theta_1\cdots d\theta_{n-1}.\\
 & = \int_0^{2\pi} d\theta_{n-1} \int_0^{\pi}d \theta_{n-2}\cdots \int_0^{\pi}d\theta_1\int_0^R \frac{1}{(2\pi\sigma^2)^{n/2}}e^{-\frac{r^2}{2\sigma^2}} r^{n-1}\sin ^{n-2}\theta_1\sin ^{n-3}\theta_2\cdots \sin \theta_{n-2}dr\\
 & = \int_0^R\frac{2\pi}{(2\pi\sigma^2)^{n/2}}e^{-\frac{r^2}{2\sigma^2}} r^{n-1}d r\int_0^{\pi}\sin ^{n-2}\theta_1d\theta_1\int_0^{\pi}\sin ^{n-3}\theta_2d\theta_2\cdots\int_0^{\pi}\sin \theta_{n-2}d\theta_{n-2}\\
 & = \int_0^R\frac{2\pi}{(2\pi\sigma^2)^{n/2}}e^{-\frac{r^2}{2\sigma^2}} r^{n-1}dr \cdot B(\frac{1}{2},\frac{n-1}{2})B(\frac{1}{2},\frac{n-2}{2})\cdots B(\frac{1}{2},1)\\
 & = \frac{[\Gamma({\frac{1}{2}})]^{n-2}}{\Gamma(\frac{n}{2})}\int_0^R\frac{2\pi}{(2\pi\sigma^2)^{n/2}}e^{-\frac{r^2}{2\sigma^2}} r^{n-1}dr\\
 & = \frac{\pi^{n/2}}{\Gamma(\frac{n}{2})}\int_0^R\frac{2}{(2\pi\sigma^2)^{n/2}}e^{-\frac{r^2}{2\sigma^2}} r^{n-1}dr\\
\end{split}
\end{equation}
% Restore the current equation number.
%\setcounter{equation}{9}
% IEEE uses as a separator
\hrulefill
% The spacer can be tweaked to stop underfull vboxes.
\vspace*{4pt}
\end{figure*}
\setcounter{equation}{137}%

In (\ref{integration21}), the function $B(x,y)$ denotes the well known Beta function.
If $x_i$ with $i = 1,\cdots,n$ follow i.i.d Gaussian distribution with zero mean and variance $\sigma^2$, denote $X = x_1^2 + \cdots + x_n^2$, then the random variable $X$ follows central $\chi$ distribution, the pdf of $X$ is written as
\begin{equation}\notag
p(x) =
\begin{cases}
\frac{1}{2^{n/2}\Gamma(n/2)\sigma^n}x^{\frac{n}{2}-1}e^{-\frac{x}{2\sigma^2}}, \ \ \ x > 0 \\
0  \ \ \ \ \ \ \ \ \ \ \ \ \ \ \ \ \ \ \ \ \ \ \ \ \ \ \ \ \ \ \ \ \   else
\end{cases}
\end{equation}
The cdf of $X$ is following when $n = 2m$ is even:
\begin{equation}\notag
F(x) =
\begin{cases}
1 - e^{-\frac{x}{2\sigma^2}}\sum_{k=0}^{m-1}\frac{1}{k!}(\frac{x}{2\sigma^2})^2, \ \ \ x > 0 \\
0  \ \ \ \ \ \ \ \ \ \ \ \ \ \ \ \ \ \ \ \ \ \ \ \ \ \ \ \ \ \ \ \ \   else
\end{cases}
\end{equation}

Note that $X = x_1^2 + \cdots + x_n^2$, we have the following equation from (\ref{integration21}),
\begin{equation}\label{integration21}
\begin{split}
 &\frac{\pi^{n/2}}{\Gamma(\frac{n}{2})}\int_0^R\frac{2}{(2\pi\sigma^2)^{n/2}}e^{-\frac{r^2}{2\sigma^2}} r^{n-1}dr\\
&\overset{r^2 = x}{\Longrightarrow}\frac{\pi^{n/2}}{\Gamma(\frac{n}{2})}\frac{2}{(2\pi\sigma^2)^{n/2}}\int_0^{R^2}e^{-\frac{x}{2\sigma^2}} x^{\frac{n-1}{2}}\frac{1}{2}x^{-\frac{1}{2}}dx\\
& = \frac{1}{2^{n/2}\Gamma(n/2)\sigma^n}\int_0^{R^2}x^{\frac{n}{2}-1}e^{-\frac{x}{2\sigma^2}}dx.
\end{split}
\end{equation}

Consequently, the integration in (\ref{ing1}) can be reformulated as
\begin{equation}
P\{X < R^2\} = \frac{1}{2^{n/2}\Gamma(n/2)\sigma^n}\int_0^{R^2}x^{\frac{n}{2}-1}e^{-\frac{x}{2\sigma^2}}dx.
\end{equation}

Denote $Y$ and $X$ as the random variable corresponding the sums of i.i.d Gaussian random variable with variance $\sigma_1$ and $\sigma$, respectively, then the equation (\ref{inital}) can be rewritten as following
\begin{equation}\label{integration2}
\begin{split}
&V_T(\mathbb{P}_0,\mathbb{P}_1)\\
 = &\frac{1}{2}\|\mathbb{P}_1(\mathbf{x}) - \mathbb{P}_0(\mathbf{x})\|_1\\
= &P\{X < R^2\} - P\{Y < R^2\}\\
= &\frac{1}{2^{\frac{n}{2}}\Gamma(n/2)}\int_0^{R^2}\left(\frac{1}{\sigma^n}x^{\frac{n}{2}-1}e^{-\frac{x}{2\sigma^2}}- \frac{1}{\sigma_1^n}x^{\frac{n}{2}-1}e^{-\frac{x}{2\sigma_1^2}}\right)dx\\
= &\frac{1}{2^{\frac{n}{2}}\Gamma(n/2)}\int_0^{R^2}x^{\frac{n}{2}-1}\left(\frac{1}{\sigma^n}e^{-\frac{x}{2\sigma^2}}- \frac{1}{\sigma_1^n}e^{-\frac{x}{2\sigma_1^2}}\right)dx.
\end{split}
\end{equation}

Now we consider the integration
\begin{equation}
\int_0^{R^2} \frac{1}{\sigma^n}x^{\frac{n}{2}-1}e^{-\frac{x}{2\sigma^2}}dx.
\end{equation}
Denote $\frac{x}{2\sigma^2} = t$, then we get
\begin{equation}
\begin{split}
&\int_0^{R^2} \frac{1}{\sigma^n}x^{\frac{n}{2}-1}e^{-\frac{x}{2\sigma^2}}dx\\
&=\frac{1}{\sigma^n}\int_0^{R^2}(2\sigma^2t)^{\frac{n}{2}-1}e^{-t}dx\\
&=\frac{1}{\sigma^n}\int_0^{R^2/2\sigma^2}2^{\frac{n}{2}-1}\sigma^{n-2}t^{\frac{n}{2}-1}e^{-t}2\sigma^2dt\\
&=\frac{1}{\sigma^n}\int_0^{R^2/2\sigma^2}2^{\frac{n}{2}}\sigma^{n}t^{\frac{n}{2}-1}e^{-t}dt\\
&=2^{\frac{n}{2}}\int_0^{R^2/2\sigma^2}t^{\frac{n}{2}-1}e^{-t}dt\\
&=2^{\frac{n}{2}}\gamma(\frac{n}{2},\frac{R^2}{2\sigma^2})
\end{split}
\end{equation}
where $\gamma(a,z) = \int_0^ze^{-t}t^{a-1}dt$ is the incomplete gamma function.

By the same reasoning, we have
\begin{equation}
\begin{split}
&\int_0^{R^2} \frac{1}{\sigma_1^n}x^{\frac{n}{2}-1}e^{-\frac{x}{2\sigma_1^2}}dx\\
=&2^{\frac{n}{2}}\gamma(\frac{n}{2},\frac{R^2}{2\sigma_1^2}).
\end{split}
\end{equation}
Therefore, the integration in (\ref{integration2}) is expressed as
\begin{equation}\label{integration3}
\begin{split}
V_T(\mathbb{P}_1,\mathbb{P}_0)= &\frac{1}{2}\|\mathbb{P}_1(\mathbf{x}) - \mathbb{P}_0(\mathbf{x})\|_1\\
= & \frac{1}{2^{\frac{n}{2}}\Gamma(n/2)} 2^{\frac{n}{2}}\left[\gamma(\frac{n}{2},\frac{R^2}{2\sigma^2})-\gamma(\frac{n}{2},\frac{R^2}{2\sigma_1^2})\right]\\
=& \frac{1}{\Gamma(n/2)}\left[\gamma(\frac{n}{2},\frac{R^2}{2\sigma^2})-\gamma(\frac{n}{2},\frac{R^2}{2\sigma_1^2})\right]\\
=& \frac{1}{\Gamma(n/2)} \int_{\frac{1}{2}n\sigma^2 \ln(1+ \frac{p_n}{\sigma^2})/p_n}^{\frac{1}{2}n\sigma_1^2 \ln(1+ \frac{p_n}{\sigma^2})/p_n}e^{-t}t^{n/2-1}dt\\
=& \frac{1}{\Gamma(n/2)}\left[\gamma(\frac{n}{2},f(\theta_n))-\gamma(\frac{n}{2},g(\theta_n))\right].
\end{split}
\end{equation}
Note that the Gamma function is related to the incomplete gamma function by $\Gamma(n/2) = \gamma(n/2, \infty)= \int_0^{\infty}e^{-t}t^{\frac{n}{2}-1}dt$.
As $R^2= \frac{n(\sigma_1\sigma)^2\ln(1 + \frac{p_n}{\sigma^2})}{p_n}$, if we denote $\theta_n = \frac{p_n}{\sigma^2}$, i.e., $\bm{snr}$, $f(\theta_n) = \frac{R^2}{2\sigma^2}$ and $g(\theta_n) = \frac{R^2}{2\sigma_1^2}$, we have the following relationships between these variables,
\begin{equation}\label{f}
\begin{split}
f(\theta_n) =&\frac{1}{2}n\sigma_1^2 \ln(1+ \frac{p_n}{\sigma^2})/p_n\\
 = &\frac{1}{2}n\frac{p_n+ \sigma^2}{p_n}\ln(1+ \frac{p_n}{\sigma^2})\\
 = &\frac{1}{2}n\left(1+ \frac{1}{\theta_n}\right)\ln(1+ \theta_n), \\
\end{split}
\end{equation}
\begin{equation}\label{g}
\begin{split}
g(\theta_n) = &\frac{R^2}{2\sigma_1^2} = \frac{1}{2}n\sigma^2 \ln(1+ \frac{p_n}{\sigma^2})/p_n\\
 = &\frac{1}{2}n\frac{\ln(1+ \theta_n)}{\theta_n}, \\\
 \end{split}
 \end{equation}
 \begin{equation}
 f(\theta_n) - g(\theta_n) = \theta_n g(\theta_n),
\end{equation}
\begin{equation}
\frac{f(\theta_n)}{g(\theta_n)} = 1 + \theta_n.
\end{equation}
\section{Proof of Theorem \ref{approximation}}\label{APPend3}

 This proof consists three steps. At first the numerical relationship between $f(\theta_n)$ and $g(\theta_n)$ is discussed, and then clarify their roles in the expansions of the incomplete gamma functions. At last we get different expansions for TVD in different cases.
From the equations (\ref{as}), (\ref{f}) and (\ref{g}),
\begin{eqnarray}\label{ftheta_n}\notag
f(\theta_n) \sim &\frac{1}{2}n\left(1+ \frac{1}{\theta_n}\right)(\theta_n - \frac{\theta_n^2}{2}+ \frac{\theta_n^3}{3} + O(\theta_n^4) )\notag\\
\sim & \frac{1}{2}n[1 + \frac{\theta_n}{2} - \frac{1}{6}\theta_n^2 + O(\theta_n^3)] > \frac{1}{2}n,
\end{eqnarray}

\begin{eqnarray}\label{gtheta_n}\notag
g(\theta_n) \sim &\frac{1}{2}n\frac{\theta_n - \frac{\theta_n^2}{2}+ \frac{1}{3}\theta_n^3 + O(\theta_n^4)}{\theta_n}\notag\\
\sim & \frac{1}{2}n[1 - \frac{\theta_n}{2}+ \frac{1}{3}\theta_n^2 + O(\theta_n^3)] < \frac{1}{2}n.
\end{eqnarray}
In addition, the following equations are obvious,
\begin{equation}\label{ft1}
\begin{split}
f(\theta_n) -\frac{n}{2} &= \frac{1}{2}n(1 + \frac{1}{\theta_n})\ln(1 + \theta_n)-\frac{n}{2}\\
&= \frac{1}{2}n[\frac{\theta_n}{2} - \frac{1}{6}\theta_n^2 + \frac{1}{12}\theta_n^3 + \cdots]\\
&= \frac{1}{2}n\times x
\end{split}
\end{equation}
where $x = \frac{\theta_n}{2} - \frac{1}{6}\theta_n^2 + \frac{1}{12}\theta_n^3 + \cdots = \sum_{j=1}^{\infty}(-1)^{j+1}(\frac{1}{j}-\frac{1}{j+1})\theta_n^j \rightarrow 0$ with $n\rightarrow \infty$.
\begin{equation}
\begin{split}
g(\theta_n)  -\frac{n}{2} &= \frac{1}{2}n\frac{\ln(1 + \theta_n)}{\theta_n} - \frac{n}{2}\\
&=  \frac{1}{2}n[ - \frac{\theta_n}{2}+ \frac{1}{3}\theta_n^2 - \frac{1}{4}\theta_n^3 + \cdots] \\
& = -\frac{1}{2}n\times y
\end{split}
\end{equation}
where $y =  \frac{\theta_n}{2} - \frac{1}{3}\theta_n^2 + \frac{1}{4}\theta_n^3 + \cdots = \sum_{j=1}^{j=\infty} (-1)^{j+1}\frac{1}{j+1}\theta_n^j\rightarrow 0$ with $ n\rightarrow \infty$.
\begin{equation}
\begin{split}
\frac{f(\theta_n)} {\frac{n}{2}} &= \frac{1}{2}n(1 + \frac{1}{\theta_n})\ln(1 + \theta_n)/\frac{n}{2}\\
&= 1 + \frac{\theta_n}{2} - \frac{1}{6}\theta_n^2 + \frac{1}{12}\theta_n^3 + \cdots\\
&= 1 + x
\end{split}
\end{equation}
\begin{equation}\label{gt1}
\begin{split}
\frac{g(\theta_n)}{\frac{n}{2} }&= \frac{\frac{1}{2}n\frac{\ln(1 + \theta_n)}{\theta_n} }{ \frac{n}{2}}\\
&=  1 - \frac{\theta_n}{2}+ \frac{1}{3}\theta_n^2 - \frac{1}{4}\theta_n^3 + \cdots \\
& = 1 - y.
\end{split}
\end{equation}
From the above analysis, we have
\begin{equation}\label{eq4}
-\frac{1}{2}n(x + y) = g(\theta_n) -f(\theta_n) =  -\frac{1}{2}n\ln(1 + \theta_n),
\end{equation}
\begin{equation}\label{gt2}
\frac{1 + x}{1 - y} = \frac{f(\theta_n)}{g(\theta_n)} = 1 + \theta_n.
\end{equation}

Now let $a = \frac{n}{2}-1$, and $z$ equals $f(\theta_n)$ and $g(\theta_n)$, respectively.
We have the following facts,
\begin{enumerate}
 \item $f(\theta_n)$ and $g(\theta_n)$ are on the right and left side of $a =\frac{n}{2}-1$ on $\mathcal{R}$, respectively.
\item Given $\theta_n$, $f(\theta_n) -(\frac{n}{2}-1)$ and $g(\theta_n) -(\frac{n}{2}-1)$ tend to $-\infty$ and $\infty$, respectively if $n\rightarrow \infty$, which implies that we can approximate them by (\ref{gammaexpan1}) and (\ref{Gammaexpan2}) when $n$ is large in case 1 and case 2 respectively. The premise condition for the above two expansions is that
    $ \left|a - z\right| = O(a^{1/2 + \epsilon})$, which implies the exponent of $ n\theta$ should satisfy
    $$1- \tau \geq \frac{1}{2},$$ that is $$\tau \leq \frac{1}{2}.$$
\item When $\theta_n$ is small with a given $n$, from (\ref{ftheta_n}) and (\ref{gtheta_n}), $f(\theta_n)$ and $g(\theta_n)$ will be very close to $\frac{n}{2}-1$ , which implies that we can approximate them by (\ref{Gammaexpan3}). Note that the premise condition is that
    $ a - z = o(a^{2/3})$, which implies the exponent of $ n\theta$ should satisfy
    $$1- \tau \leq \frac{2}{3},$$ that is $$\alpha \geq \frac{1}{3}.$$ Hence, $V_T(\mathbb{P}_1,\mathbb{P}_0)$ could be approximated by the expansions from (\ref{Gammaexpan3}) if $\tau > \frac{1}{2}$.
\end{enumerate}
Now we consider the expansions for $V_T(\mathbb{P}_1,\mathbb{P}_0)$ when $\tau \geq \frac{1}{2}$ and $\tau < \frac{1}{2}$, respectively.
First, from (\ref{Gamma0}),
\begin{equation}\label{Gammaexpan}
n! \sim e^{-n}n^n\sqrt{2\pi n}
\end{equation}
By Legendre's duplication formula,
\begin{equation}\label{Gamma}
\sqrt{\pi}\Gamma(2z) = 2^{2z-1}\Gamma(z)\Gamma(z+\frac{1}{2})
\end{equation}
In our setting, $z$ is a integer $n$ , hence
$$\Gamma(n+\frac{1}{2}) = \dbinom{n-\frac{1}{2}}{n}\Gamma(n)\sqrt{\pi} \sim \sqrt{\pi}\cdot\Gamma(n).$$
 Therefore from (\ref{Gamma})
\begin{equation}\label{Gamman2}
\begin{split}
& 2^{n-1}\sqrt{\pi}\Gamma(\frac{n}{2})^2 \sim e^{-n}n^n\sqrt{2\pi n}\cdot\sqrt{\pi}\\
&\Rightarrow \Gamma(\frac{n}{2}) \sim e^{-\frac{n}{2}}(\frac{n}{2})^{\frac{n}{2}}n^{\frac{1}{4}}\sqrt{\pi}\cdot2^{\frac{5}{4}}
\end{split}
\end{equation}
the detailed expansion for our approximation of $V_T(\mathbb{P}_0,\mathbb{P}_1)$ when $\tau \geq \frac{1}{2}$.
\begin{equation}\label{alphageq}
\begin{split}
&V_T(\mathbb{P}_0,\mathbb{P}_1)= \\
&= \frac{1}{\Gamma(n/2)}\left[\gamma(\frac{n}{2},f(\theta_n))-\gamma(\frac{n}{2},g(\theta_n))\right]\\
&\overset{(a)}{=}\frac{1}{\Gamma(n/2)}\left[\Gamma(\frac{n}{2},g(\theta_n))-\Gamma(\frac{n}{2},f(\theta_n))\right]\\
&\sim \frac{ e^{-a}a^{a+1}\sum_{k=0}^{\infty}c_k(a)\left[\Phi_k(a,g(\theta)) -  \Phi_k(a,f(\theta))\right]}{\Gamma(a+1)}\\
&\overset{b}{\sim} \frac{ e \cdot(1- \frac{1}{a+1})^{a+1}}{(2a+2)^{\frac{1}{4}}\sqrt{\pi}\cdot2^{\frac{5}{4}}}\sum_{k=0}^{\infty}c_k(a)\left[\Phi_k(a,g(\theta)) -  \Phi_k(a,f(\theta))\right]\\
&\overset{(c)}{\sim} \frac{1}{n^{\frac{1}{4}}\sqrt{\pi}\cdot2^{\frac{5}{4}}}\sum_{k=0}^{\infty}c_k(a)\left[\Phi_k(a,g(\theta)) -  \Phi_k(a,f(\theta))\right]\\
\end{split}
\end{equation}
 where (a) is from $\Gamma(a,z) = \Gamma(a) - \gamma(a,z)$, (b) is from (\ref{Gamman2}), $\Gamma(a+1) = \Gamma(\frac{n}{2}) \sim e^{-\frac{n}{2}}(\frac{n}{2})^{\frac{n}{2}}(n)^{\frac{1}{4}}\sqrt{\pi}\cdot2^{\frac{5}{4}} = e^{-a-1}(a+1)^{a+1}(2a+2)^{\frac{1}{4}}\sqrt{\pi}\cdot2^{\frac{5}{4}}$ and (c) is from $ \lim_{a \to \infty}(1- \frac{1}{a+1})^{a+1} =  e^{-1}$.

 When $\tau < \frac{1}{2}$, $V_T(\mathbb{P}_1,\mathbb{P}_0)$ could be rewritten as
\begin{equation}
\begin{split}
&\frac{1}{\Gamma(\frac{n}{2})}\left[\gamma(\frac{n}{2},f(\theta_n)) - \gamma(\frac{n}{2},g(\theta_n))\right]\\
= &\frac{1}{\Gamma(\frac{n}{2})}\left[\Gamma(\frac{n}{2})- \Gamma(\frac{n}{2},f(\theta_n))-\gamma(\frac{n}{2},g(\theta_n))\right]\\
=& 1-\frac{1}{\Gamma(\frac{n}{2})}\left[ \Gamma(\frac{n}{2},f(\theta_n))+\gamma(\frac{n}{2},g(\theta_n))\right].
\end{split}
\end{equation}

We have the following asymptotic expansion for $\Gamma(\frac{n}{2},f(\theta_n))+\gamma(\frac{n}{2},g(\theta_n))$:
\begin{equation}\label{expansion1}
\begin{split}
&\frac{1}{\Gamma(\frac{n}{2})}\left[ \Gamma(\frac{n}{2},f(\theta_n))+\gamma(\frac{n}{2},g(\theta_n))\right]\\
\sim \,\ \,&  e^{-f(\theta_n)+ \frac{n}{2}}\left(\frac{f(\theta_n)}{\frac{n}{2}}\right)^{\frac{n}{2}}\frac{1}{\sqrt{\pi}n^{\frac{1}{4}}}\sum_{k=0}^{\infty}\frac{ (-1)^k k! c_k(\frac{n}{2}-1)}{(f(\theta_n)+1 -\frac{n}{2})^{k+1}} \\
 + \,\ \, &  e^{-g(\theta_n)+ \frac{n}{2}}\left(\frac{g(\theta_n)}{\frac{n}{2}}\right)^{\frac{n}{2}}\frac{1}{\sqrt{\pi}n^{\frac{1}{4}}}\sum_{k=0}^{\infty}\frac{ (-1)^{k+1} k! c_k(\frac{n}{2}-1)}{(g(\theta_n)+1 -\frac{n}{2})^{k+1}}\\
 + \,\ \,&e^{-g(\theta_n)+ \frac{n}{2}}\left(\frac{g(\theta_n)}{\frac{n}{2}}\right)^{\frac{n}{2}}\frac{1}{\sqrt{\pi}n^{\frac{1}{4}}}\sum_{k=0}^{\infty}c_k(\frac{n}{2}-1) e^{g(\theta_n)+1 -\frac{n}{2}}\\
 &\cdot\sum_{j=0}^k\frac{(-1)^j k!}{(k-j)!(g(\theta_n)+1 -\frac{n}{2})^{j+1}}.\\
 \end{split}
 \end{equation}
For the terms $e^{-f(\theta_n)+ \frac{n}{2}}\left(\frac{f(\theta_n)}{\frac{n}{2}}\right)^{\frac{n}{2}}$ and $e^{-g(\theta_n)+ \frac{n}{2}}\left(\frac{g(\theta_n)}{\frac{n}{2}}\right)^{\frac{n}{2}}$, with the help of (\ref{ft1}) - (\ref{gt2}), we have
\begin{equation}
\begin{split}
 &\frac{e^{-f(\theta_n)+ \frac{n}{2}}\left(\frac{f(\theta_n)}{\frac{n}{2}}\right)^{\frac{n}{2}}}{e^{-g(\theta_n)+ \frac{n}{2}}\left(\frac{g(\theta_n)}{\frac{n}{2}}\right)^{\frac{n}{2}}}\\
= &\frac{e^{-\frac{1}{2}nx}\left(1 + x \right)^{\frac{n}{2}} }{  e^{\frac{1}{2}ny}\left(1 - y\right)^{\frac{n}{2}}} \\
=& e^{-\frac{1}{2}n(x + y)}\left(\frac{1+x}{1-y}\right)^{\frac{n}{2}} \\
= & 1.
\end{split}
\end{equation}

\end{document}